\documentclass[letterpaper,11pt]{article}

\usepackage{amsfonts}
\usepackage{amsmath}
\usepackage{amsthm}
\usepackage{amssymb}
\usepackage{mathtools}
\usepackage[hyphens]{url}
\usepackage{hyperref}
\usepackage{graphicx}
\usepackage{subfig}

\newtheorem{theorem}{Theorem}

\newtheorem{lemma}{Lemma}

\begin{document}
\title{Shape Formation by Programmable Particles}
\author{Giuseppe A.~Di Luna\footnotemark[1] \and Paola Flocchini\footnotemark[1] \and Nicola Santoro\footnotemark[2] \and Giovanni Viglietta\footnotemark[1] \and Yukiko Yamauchi\footnotemark[3]}
\date{}
\thispagestyle{empty}
\renewcommand{\thefootnote}{\fnsymbol{footnote}}
\footnotetext[1]{University of Ottawa, Canada. E-mails: \texttt{\{gdiluna, paola.flocchini, gvigliet\}@uottawa.ca}.}
\footnotetext[2]{Carleton University, Canada. E-mail: \texttt{santoro@scs.carleton.ca}.}
\footnotetext[3]{Kyushu University, Japan. E-mail: \texttt{yamauchi@inf.kyushu-u.ac.jp}.}
\renewcommand{\thefootnote}{\arabic{footnote}}
\setcounter{footnote}{0}

\maketitle

\begin{abstract}
\emph{Shape formation} (or \emph{pattern formation})  is a basic distributed  problem
for  systems of computational  mobile entities.
Intensively studied  for  systems  of
autonomous mobile robots, it has recently been investigated in the realm of  {\em programmable matter}, 
where entities are assumed to be  small and  with severely limited capabilities.
Namely, it  has been studied
in the geometric  {\em  Amoebot} model, where
 the anonymous entities, called {\em particles}, 
 operate on 
a  hexagonal tessellation of the plane  and have
limited  computational power (they have constant memory), strictly local interaction and communication capabilities (only with  particles in neighboring nodes of the  grid), 
 and limited motorial capabilities (from a grid node to an empty neighboring node);  their activation is controlled
 by an adversarial   scheduler.
Recent investigations  have  shown    how,
starting  from a well-structured configuration in which the particles form a (not necessarily complete) triangle,
the particles can form a large  class  of shapes.
This result
has been established under several assumptions: 
agreement on the clockwise direction (i.e., {\em chirality}),
a {\em sequential} activation schedule,
and {\em randomization} (i.e., particles can flip coins to elect a leader).

In this paper 
 we obtain several results that, among other things,   provide a   characterization of 
which shapes can be  formed {\em deterministically}  starting from any   {\em simply connected}  initial configuration of $n$ particles.
The characterization is constructive: we provide a
\emph{universal shape formation algorithm} that, for each feasible pair of shapes $(S_0, S_F)$, allows the
particles to form the final shape $S_F$ (given in input) starting from  
 the  initial shape $S_0$, unknown to the particles. The final configuration will be an appropriate scaled-up copy of $S_F$ depending on $n$.
 
If  {\em randomization} is allowed,   then  any input shape  can  be formed  from any initial (simply connected) shape
 by our algorithm,  provided that there are enough  particles.
 
Our algorithm works   without  chirality, proving that
 {\em chirality  is   computationally  irrelevant} for shape formation. Furthermore, it works
 under a strong adversarial scheduler, not necessarily sequential.

We also consider the complexity of shape formation both in terms of the number of rounds and the total number of moves performed by the particles executing a universal shape formation algorithm. We prove that our solution has a complexity of $O(n^2)$ rounds and moves: this number of moves is also asymptotically optimal.
\end{abstract}

\section{Introduction}\label{s:1}

\subsection{Background}
The term {\em  programmable matter}, introduced by Toffoli and Margolus over a quarter century ago~\cite{ToM91},
 is used to denote  matter that has the ability to change its physical properties (e.g., shape,  color, density, etc.)
 in a programmable fashion, based upon user input or autonomous sensing. 
  Often programmable matter is envisioned  as a very large number of very small locally interacting computational particles,
  programmed to  collectively perform a complex task. Such particles  could have applications 
 in a variety of  important situations: they
  could be employed to create smart materials,  used for autonomous monitoring and repair, 
be instrumental in  minimal invasive surgery, etc.  

As  recent advances in microfabrication and cellular engineering render the production of such particles increasingly possible,
there has been a convergence of theoretical research interests
on programmable matter from some areas of computer science, 
 especially robotics, sensor networks, molecular self-assembly, and  distributed computing.
Several theoretical models  of programmable matter have been proposed, ranging from DNA self-assembly systems (e.g.,~\cite{DemPSS11,michail1,michail2,Pa14,Ro06,ScW15})
to shape-changing synthetic molecules and  cells (e.g.,~\cite{WoCGD+13}), from metamorphic robots (e.g.,~\cite{Chi94,WaWA04}) to  nature-inspired synthetic
insects and micro-organisms (e.g.,~\cite{DerGSB+15,DoFRNV+16,LiTTRS10}), each model assigning  special capabilities and constraints to the entities  
and focusing on specific  applications.

Of particular interest,  from the distributed computing viewpoint,  is the  {\em geometric Amoebot} model
of programmable matter~\cite{CaDRR16,DaDGP+17,DaGPR+17,DerGMR+15,spaa,DerGMR+17,DerGSB+15}.
 In this model, introduced in~\cite{DerGSB+15} and so called because inspired by the behavior of amoeba,
programmable matter  is viewed as a swarm of decentralized autonomous self-organizing  entities, operating on 
a  hexagonal tessellation of the plane. 
These entities, called {\em particles},  are constrained by having simple computational capabilities (they are finite-state
 machines), strictly local interaction and communication capabilities (only with  particles located in neighboring nodes of the hexagonal grid), 
 and limited motorial capabilities (from a grid node to an empty neighboring node); furthermore, their activation is controlled
 by an adversarial (but  fair) synchronous scheduler.
 A feature of the Amoebot model is that particles can be in two modes: {\em contracted} and 
{\em expanded}. When contracted, a
 particle occupies only one node, while when expanded the particle occupies two neighboring nodes;
 it is indeed this ability of a particle  to expand and contract that allows it to move on the grid.
 The Amoebot model has been  investigated to understand the computational power of such simple entities;
  the focus has been  on applications such as
 coating~\cite{DaDGP+17,DerGMR+17}, gathering~\cite{CaDRR16},  and   shape formation~\cite{DerGMR+15,spaa,DerGSB+15}. The latter is also the topic of our investigation.

The \emph{shape formation} problem is  prototypical   for   systems of self-organizing  entities.
This problem, called \emph{pattern formation} in swarm robotics, requires the  entities to move in 
the spatial universe they inhabit in such a way that,
within finite time, their positions form the  geometric shape given in input  (modulo translation, rotation, scaling, and reflection),
and no further changes occur. 
Indeed, this problem has been intensively studied especially in active systems  such as 
autonomous mobile robots (e.g.,~\cite{AnSY95,DaFSY15,FlPSW08,FuYKY16,SuzY99,YaS10}) and 
modular robotic systems (e.g.,~\cite{ArR10,newref,RuCN14}). 

In the  Amoebots model, shape formation has been investigated in~\cite{DerGMR+15,spaa,DerGSB+15},
taking into account that, due to the ability of particles to expand, 
it might be possible to form shapes  whose    size is  larger than
the number of particles. 

The pioneering study of~\cite{DerGMR+15} on shape formation in the geometric Amoebot model showed 
how particles can build simple shapes, such as a hexagon or a triangle. 
Subsequent investigations have recently shown      how,
starting  from a well-structured configuration in which the particles form a (not necessarily complete) triangle,
they can form a larger class  of shapes~\cite{spaa}. 
This result
has been established under several assumptions: 
availability of {\em chirality}
(i.e., a globally consistent circular orientation of the plane shared by all particles), 
a {\em sequential} activation schedule (i.e., at each time unit the scheduler selects only one particle which will interact with its
neighbors and possibly move),
and, more important,  {\em randomization} (i.e., particles can flip coins to elect a leader).

These results  and assumptions immediately and naturally open fundamental research questions, including:
Are other shapes formable?  What can be done deterministically?  Is chirality necessary?
as well as  some less crucial but nevertheless  interesting questions, such as: 
What happens if the scheduler is not sequential? What if
the initial configuration is not well structured?

In this paper, motivated and stimulated by these questions, we continue the investigation  on shape formation in the geometric 
Amoebot model
and provide some definitive answers.

\subsection{Main Contributions}

We establish several results that, among other things,  provide a  constructive characterization of 
which shapes $S_F$ can be  formed {\em deterministically}  starting from an unknown   {\em simply connected}  initial configuration $S_0$ of $n$ particles.

As in~\cite{spaa}, we assume that the size of the description of  $S_F$ is constant with respect 
to the size of the system, so that it can be encoded by each particle as part of its internal memory. Such a description is available to all the particles at the beginning of the execution, and we call it their ``input''. The particles will form a final configuration that is an appropriate scaling, translation, rotation, and perhaps reflection of the input shape $S_F$. Since all particles of $S_0$ must be used to construct $S_F$, the scale $\lambda$ of the final configuration depends on $n$: we stress that $\lambda$ is unknown to particles, and they must determine it autonomously. 

Given  two shapes  $S_0$ and $S_F$, we say that  the pair $(S_0,S_F)$ is {\em feasible} if there exists a deterministic algorithm that, in every execution and
regardless of the activation schedule,  allows the particles to form $S_F$ starting from $S_0$ and no longer move.

On the contrary, a pair $(S_0,S_F)$ of shapes is {\em unfeasible} when the symmetry of the initial configuration $S_0$ prevents the formation of the final shape $S_F$. In Section~\ref{s:2}, we formalize the notion of {\em unbreakable} symmetry of shapes embedded in triangular grids, and in  Theorem~\ref{t:neg} we show that starting from an unbreakable $k$-symmetric configuration only unbreakable $k$-symmetric shapes can be formed.

Interestingly, for all the feasible pairs,  we provide a \emph{universal shape formation algorithm} in Section~\ref{s:3}. This algorithm does not need any information on $S_0$, except that it is simply connected.

These results concern the \emph{deterministic} formation of shapes. As a matter of fact, our algorithm uses a deterministic leader election algorithm as a subroutine (Sections~\ref{s:3.1}--\ref{s:3.4}).
If the initial shape $S_0$ is unbreakably $k$-symmetric, such an algorithm may elect as many as $k$ neighboring leader particles, where $k \in \{1,2,3\}$. It is trivial to see that, with a constant number of coin tosses, we can elect a unique leader among these $k$ with arbitrarily high probability.
Thus, our results immediately imply the existence of a \emph{randomized} universal shape formation algorithm for \emph{any} pair of shapes $(S_0,S_F)$ where $S_0$ is simply connected. This extends the result of~\cite{spaa}, which assumes the initial configuration to be a (possibly incomplete) triangle.

Additionally, our notion of shape generalizes the one used in~\cite{spaa}, where a shape is only a collection of triangles, while we include also 1-dimensional segments as its constituting elements. In Section~\ref{s:4}, we are going to show how the concept of shape can be further generalized to essentially include anything that is Turing-computable.

Our algorithm works 
 under a  stronger adversarial scheduler 
 that activates an arbitrary   number of particles at each stage (i.e., not necessarily just one, like the sequential scheduler),
and with a slightly less demanding 
communication system.

Moreover, in our  algorithm  no chirality is assumed: indeed, unlike in~\cite{spaa},  different particles may have different handedness. On the contrary, in the examples of unfeasibility given in Theorem~\ref{t:neg}, all particles have the same handedness. Together, these two facts allows us to conclude that {\em chirality  is   computationally  irrelevant} for shape formation.

Finally, we analyze the complexity  of shape formation in terms of  the total number of {\em moves} (i.e., contractions and expansions) 
performed by $n$ particles 
executing a universal shape formation algorithm, as well as in terms of the total number of \emph{rounds} (i.e., spans of time in which each particle is activated at least once, also called \emph{epochs}) taken by the particles. 
We first prove that any  universal shape formation algorithm requires  $\Omega(n^2)$ moves in the worst case (Theorem~\ref{t:lower}). We then show that
the total number of moves of our  algorithm  is $O(n^2)$ in the worst case (Theorem~\ref{tfinal}): that is, our solution is asymptotically optimal. The time complexity of our algorithm is also $O(n^2)$ rounds, and optimizing it is left as an open problem (we are able to reduce it to $O(n\log n)$, and we have a lower bound of $\Omega(n)$: see Section~\ref{s:4}).

Obviously, we must assume the size of $S_0$ (i.e., the number of particles that constitute it) to be sufficiently large with respect to the input description of the final shape $S_F$. More precisely, denoting the size of $S_F$ as $m$, we assume $n$ to be lower-bounded by a cubic function of $m$ (Theorem~\ref{tp6}). A similar restriction is also found in~\cite{spaa}.

\section{Model and Preliminaries}\label{s:2}

\textbf{Particles.}
A \emph{particle} is a conceptual model for a computational entity that lives in an abstract graph $G$. A particle may occupy either one vertex of $G$ or two adjacent vertices: in the first case, the particle is said to be \emph{contracted}; otherwise, it is \emph{expanded}.

\smallskip
\noindent\textbf{Movement.}
A particle may move through $G$ by performing successive \emph{expansion} and \emph{contraction} operations.\footnote{The model in~\cite{spaa} allows a special type of coordinated move called ``handover''. Since we will not need our particles to perform this type of move, we omit it from our model.} Say $v$ and $u$ are two adjacent vertices of $G$, and a contracted particle $p$ occupies $v$. Then, $p$ can expand toward $u$, thus occupying both $v$ and $u$. When such an expansion occurs, $u$ is said to be the \emph{head} of $p$, and $v$ is its \emph{tail}. From this position, $p$ can contract again into its head vertex $u$. As a general rule, when a particle expands toward an adjacent vertex, this vertex is by definition the particle's head. An expanded particle cannot expand again unless it contracts first, and a contraction always brings a particle to occupy its head vertex. When a particle is contracted, the vertex it occupies is also called the particle's head; a contracted particle has no tail vertex.

If a graph contains several particles, none of its vertices can ever be occupied by more than one particle at the same time. Accordingly, a contracted particle cannot expand toward a vertex that is already occupied by another particle. If two or more particles attempt to expand toward the same (unoccupied) vertex at the same time, only one of them succeeds, chosen arbitrarily by an adversarial scheduler (see below).

\smallskip
\noindent\textbf{Scheduler.}
In our model, time is ``discrete'', i.e., it is an infinite ordered sequence of instants, called \emph{stages}, starting with stage 0, and proceeding with stage 1, stage 2, etc. Say that in the graph $G$ there is a set $P$ of particles, which we call a \emph{system}. At each stage, some particles of $P$ are \emph{active}, and the others are \emph{inactive}. We may think of the activation of a particle as an act of an adversarial \emph{scheduler}, which arbitrarily and unpredictably decides which particles are active at each stage.
The only restriction on the scheduler is a bland \emph{fairness} constraint, requiring that each particle be active for infinitely many stages in total. That is, the scheduler can never keep a particle inactive forever.

\smallskip
\noindent\textbf{Sensing and reacting.}
When a particle is activated for a certain stage, it ``looks'' at the vertices of $G$ adjacent to its head, discovering if they are currently unoccupied, or if they are head or tail vertices of some particle. If the particle is expanded, it also detects which of these vertices is its own tail; all other particles are indistinguishable (i.e., they are \emph{anonymous}). Each active particle may then decide to either expand (if it is contracted), contract (if it is expanded), or stay still for that stage. All these operations are performed by all active particles simultaneously, and take exactly one stage. So, when the next stage starts, a new set of active particles is selected, which observe their surroundings and move, and so on.

\smallskip
\noindent\textbf{Memory.}
Each particle has an \emph{internal state} that it can modify every time it is activated. The internal state of any particle must be picked from a finite set $Q$; i.e., all particles have ``finite memory''.

\smallskip
\noindent\textbf{Communication.}
Two particles can also \emph{communicate} by sending each other \emph{messages} taken from a finite set $M$, provided that their heads are adjacent vertices of $G$. Specifically, when a particle $p$ is activated and sees the head of particle $p'$, it may send a message $m$ to it along the oriented edge $(u,v)$ connecting their heads. Then, the next time $p'$ is activated, it will receive and read the message $m$. That is, unless some particle (perhaps again $p$) sends another message $m'$ on the same oriented edge $(u,v)$ before $p'$ is activated, in which case $m$ is ``overwritten'' by $m'$, unbeknownst to $p'$ and the particle that sent $m'$. In other words, upon activation, a particle will always receive the most recent message sent to it from each of the vertices adjacent to its head, while older unread messages are lost. If $p'$ expands while $p$ is sending a message to it (i.e., in the same stage), the message is lost and is not received by any particle. When a message has been read by a particle, it is immediately destroyed.\footnote{The model in~\cite{spaa} has a more demanding communication system, which assumes each particle to have some local shared memory that all neighboring particles can read and modify.}

\smallskip
\noindent\textbf{Triangular network.}
In this paper, as in~\cite{spaa}, we assume the graph $G$ to be the dual graph of a regular hexagonal tiling of the Euclidean plane. So, in the following, $G$ will be an infinite regular triangular grid. We also denote by $G_D$ a fixed ``canonical'' \emph{drawing} of the abstract graph $G$ in which each face is embedded in the Cartesian plane as an equilateral triangle of unit side length with one edge parallel to the $x$ axis.

\smallskip
\noindent\textbf{Port labeling.}
Note that each vertex of $G$ has degree 6. With each particle $p$ and each vertex $v$ is associated a \emph{port labeling} $\ell(p,v)$, which is a numbering of the edges incident to $v$, from 0 to 5, in clockwise or counterclockwise order with respect to the drawing $G_D$. For a fixed particle $p$, port labels are assumed to be invariant under the automorphisms of $G$ given by translations of its drawing $G_D$. As a consequence, if the port labeling $\ell(p,v)$ assigns the label $i$ to the edge $(v,u)$, then the port labeling $\ell(p,u)$ assigns the label $(i+3)\ \rm{mod}\ 6$ to the edge $(u,v)$. However, different particles may have different port labels for the same vertex $v$, depending on what edge (incident to $v$) is assigned the label 0, and whether the labels proceed in clockwise or counterclockwise order around $v$. If they proceed in clockwise order, the particle is said to be \emph{right-handed}; otherwise, it is \emph{left-handed}. So, the \emph{handedness} of a particle does not change as the particle moves, but different particles may have different handedness.

\smallskip
\noindent\textbf{Stage structure.}
Summarizing, an active particle $p$ performs the following actions during a single stage: it reads its current internal state $q$, it looks at the contents $c_0,$ \dots $c_5$ of the vertices adjacent to its head (each $c_i$ has four possible values describing the vertex corresponding to port $i$ according to $p$'s labeling: it may denote an unoccupied vertex, $p$'s own tail, the head of another particle, or the tail of another particle), it reads the pending messages $m_0$, \dots, $m_5$ coming from the vertices adjacent to its head (again, indices correspond to port labels, and some $m_i$'s may be the empty string $\varepsilon$, denoting the absence of a message), it changes its internal state to $q'$, it sends messages $m'_0$, \ldots, $m'_5$ to the vertices adjacent to its head, which possibly replace older unread messages (if $m'_i=\varepsilon$, no message is sent through port $i$), and it performs an operation $o$ (there are eight possibilities for $o$: stay still, contract, or expand toward the vertex corresponding to some port label). These variables are related by the equation $A(q,c_0,\dots,c_5,m_0,\ldots, m_5)=(q',m'_0,\ldots, m'_5,o)$, where $A$ is a function.

Recall that the set $Q$ of possible internal states is finite, as well as the set $M$ of possible messages. Hence $A$ is a finite function, and we will identify it with the \emph{deterministic algorithm} that computes it.

We assume that, when stage 0 starts, all particles are contracted, they all have the same predefined internal state $q_0$, and there are no messages pending between particles.

\smallskip
\noindent\textbf{Shapes.}
In this paper we study shapes and how they can be formed by systems of particles. A \emph{shape} is a non-empty connected set consisting of the union of finitely many edges and faces of the drawing $G_D$.\footnote{In Section~\ref{s:4}, we will show that our results hold also for a much more general notion of shape.} We stress that a shape is not a subgraph of the abstract graph $G$, but it is a subset of $\mathbb R^2$, i.e., a geometric set. A shape $S$ is \emph{simply connected} if the set $\mathbb R^2\setminus S$ is connected (intuitively, $S$ has no ``holes''). The \emph{size} of a shape is the number of vertices of $G_D$ that lie in it.

We say that two shapes $S$ and $S'$ are \emph{equivalent} if $S'$ is obtained from $S$ by a similarity transformation, i.e., a composition of a translation, a rotation, an isotropic scaling by a positive factor, and an optional reflection. Clearly, our notion of equivalence is indeed an equivalence relation between shapes.

\begin{figure}[ht]
\begin{center}
  \subfloat{
  \includegraphics[scale=0.625]{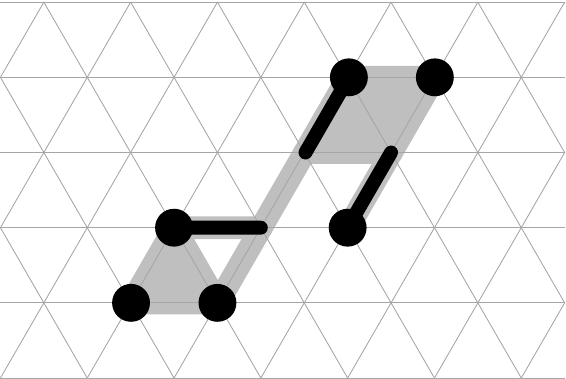}
  \label{f:shape1}
  }
  \hfill
  \subfloat{
  \includegraphics[scale=0.625]{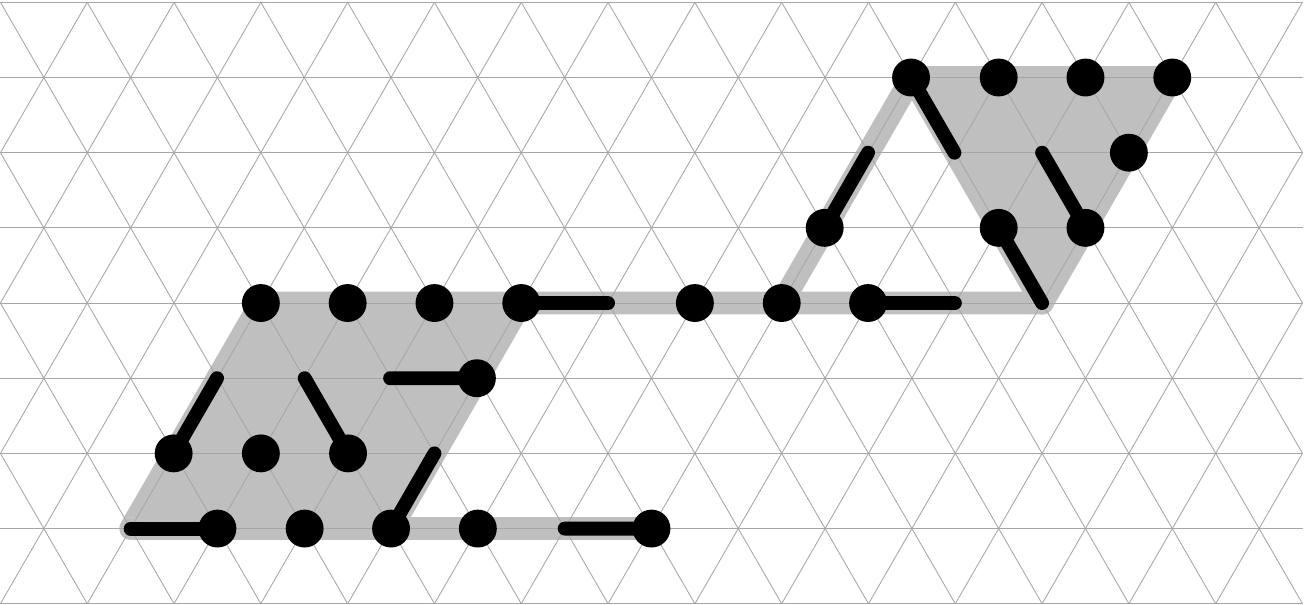}
  \label{f:shape2}
  }
\end{center}
\caption{Two systems of particles forming equivalent shapes. The shape on the left is minimal; the one on the right has scale $3$. Contracted particles are represented as black dots; expanded particles are black segments. Shapes are indicated by gray blobs.}
\label{f:shape}
\end{figure}

A shape is \emph{minimal} if no shape that is equivalent to it has a smaller size. Obviously, any shape $S$ is equivalent to a minimal shape $S'$. The size of $S'$ is said to be the \emph{base size} of $S$. Let $\sigma$ be a similarity transformation such that $S=\sigma(S')$. We say that the (positive) scale factor of $\sigma$ is the \emph{scale} of $S$.

\begin{lemma}\label{l:shapescale}
The scale of a shape is a positive integer.
\end{lemma}
\begin{proof}
Let $S$ and $S'$ be equivalent shapes, with $S'$ minimal, and let $\sigma$ be a similarity transformation with $S=\sigma(S')$. Observe that there is a unique covering of $S$ by maximal polygons and line segments with mutually disjoint relative interiors. Each of these segments and each edge of these polygons is the union of finitely many edges of $G_D$, and therefore it has integral length. It is easy to see that the scale factor of $\sigma$ must be the greatest common divisor of all such lengths, which is a positive integer.
\end{proof}

\smallskip
\noindent\textbf{Shape formation.}
We say that a system of particles in $G$ \emph{forms} a shape $S$ if the vertices of $G$ that are occupied by particles correspond exactly to the vertices of the drawing $G_D$ that lie in $S$.

Suppose that a system forms a shape $S_0$ at stage 0, and let all particles execute the same algorithm $A$ every time they are activated. Assume that there exists a shape $S_F$ such that, however the port labels of each particle are arranged, and whatever the choices of the scheduler are, there is a stage where the system forms a shape \emph{equivalent} to $S_F$ (not necessarily $S_F$), and such that no particle ever contracts or expands after that stage. Then, we say that $A$ is an \emph{$(S_0,S_F)$-shape formation algorithm}, and $(S_0,S_F)$ is a \emph{feasible} pair of shapes. (Among the choices of the scheduler, we also include the decision of which particle succeeds in expanding when two or more of them intend to occupy the same vertex at the same stage.)

In the rest of this paper, we will characterize the feasible pairs of shapes $(S_0,S_F)$, provided that $S_0$ is simply connected and its size is not too small. That is, for every such pair of shapes, we will either prove that no shape formation algorithm exists, or we will give an explicit shape formation algorithm. Moreover, all  algorithms have the same structure, which does not depend on the particular choice of the shapes $S_0$ and $S_F$. We could even reduce all of them to a single \emph{universal shape formation algorithm}, which takes the ``final shape'' $S_F$ (or a representation thereof) as a parameter, and has no information on the ``initial shape'' $S_0$, except that it is simply connected. As in~\cite{spaa}, we assume that the size of the parameter $S_F$ is constant with respect to the size of the system, so that $S_F$ can be encoded by each particle as part of its internal memory. More formally, we have infinitely many universal shape formation algorithms $A_m(S_F)$, one for each possible size $m$ of the parameter $S_F$.

Next, we will state our characterization of the feasible pairs of shapes, along with a proof that there is no shape formation algorithm for the unfeasible pairs. In Section~\ref{s:3}, we will give our universal shape formation algorithm for the feasible pairs.

\smallskip
\noindent\textbf{Unformable shapes.}
There are cases in which an $(S_0,S_F)$-shape formation algorithm does not exist. The first, trivial one, is when the size of $S_0$ is not large enough compared to the base size of $S_F$. The second is more subtle, and has to do with the fact that certain symmetries in $S_0$ cannot be broken.

A shape is said to be \emph{unbreakably $k$-symmetric}, for some integer $k>1$, if it has a center of $k$-fold rotational symmetry that does not coincide with any vertex of $G_D$. Observe that the order of the group of rotational symmetries of a shape must be a divisor of 6. However, the shapes with a 6-fold rotational symmetry have center of rotation in a vertex of $G_D$. Hence, there exist unbreakably $k$-symmetric shapes only for $k=2$ and $k=3$.

The property of being unbreakably $k$-symmetric is invariant under equivalence, provided that the scale remains the same.

\begin{lemma}\label{l:shapeequi}
A shape $S$ is unbreakably $k$-symmetric if and only if all shapes that are equivalent to $S$ and have the same scale as $S$ are unbreakably $k$-symmetric.
\end{lemma}
\begin{proof}
Observe that any point of $S$ with maximum $x$ coordinate must be a vertex of $G_D$. Let $v$ be one of such vertices. Since $S$ is connected and contains at least one edge or one face of $G_D$, there must exist an edge $uv$ of $G_D$ that lies on the boundary of $S$. Consider a similarity transformation $\sigma$ that maps $S$ into an equivalent shape $S'$ with the same scale (hence $\sigma$ is an isometry). Since $v$ is an extremal point of $S$, it must be mapped by $\sigma$ into an extremal point of $S'$, which must be a vertex $v'$ of $G_D$, as well. Analogously, $\sigma(uv)$ must be a segment located on the boundary of $S'$. The length of $\sigma(uv)$ is the length of an edge of $G_D$, and its endpoint $v'=\sigma(v)$ is a vertex of $G_D$. It follows that $\sigma(uv)$ is an edge of $G_D$. This implies that a point $p$ is a vertex of $G_D$ if and only if $\sigma(p)$ is a vertex of $G_D$.

Clearly, $S$ is rotationally symmetric if and only if $S'$ is. Suppose that $S$ has a $k$-fold rotational symmetry with center $c$, and therefore $S'$ has a $k$-fold rotational symmetry with center $c'=\sigma(c)$. By the above reasoning, $c'$ is a vertex of $G_D$ if and only if $c$ is, which is to say that $S'$ is unbreakably $k$-symmetric if and only if $S$ is.
\end{proof}

Nonetheless, if the scale changes, the property of being unbreakably $k$-symmetric may or may not be preserved. The next lemma, which extends the previous one, gives a characterization of when this happens.

\begin{lemma}\label{l:shapemin}
A shape $S$ is unbreakably $k$-symmetric if and only if any minimal shape that is equivalent to $S$ is also unbreakably $k$-symmetric, and the scale of $S$ is not a multiple of $k$.
\end{lemma}
\begin{proof}
By Lemma~\ref{l:shapescale}, the scale of $S$ is a positive integer, say $\lambda$. Let $S'$ be a minimal shape equivalent to $S$. Of course, equivalent shapes have the same group of rotational symmetries, and therefore $S'$ has a $k$-fold rotational symmetry if and only if $S$ does. We will first prove that, if $S'$ is not unbreakably $k$-symmetric, then neither is $S$. So, suppose that $S'$ has a $k$-fold rotational symmetry with center in a vertex $v$ of $G_D$. Let $\sigma$ be the homothetic transformation with center $v$ and ratio $\lambda$. Then, $S''=\sigma(S')$ is a shape with center $v$ and scale $\lambda$, which is therefore not unbreakably $k$-symmetric. Since $S$ is equivalent to $S''$ (as they are both equivalent to $S'$) and has the same scale, it follows by Lemma~\ref{l:shapeequi} that $S$ is not unbreakably $k$-symmetric, either.

Assume now that $S'$ is unbreakably $k$-symmetric. As already observed, we have two possible cases: $k=2$ and $k=3$. If $k=2$ (respectively, $k=3$), the center of symmetry $c$ of $S'$ must be the midpoint of an edge $uv$ of $G_D$ (respectively, the center of a face $uvw$ of $G_D$). Consider the homothetic transformation $\sigma$ with center $u$ and ratio $\lambda$. It is clear that $\sigma$ maps $S'$ into an equivalent shape $S''$ whose scale is $\lambda$ and whose center of symmetry is $\sigma(c)$. As Figure~\ref{f:unbreakable} suggests, $\sigma(c)$ is a vertex of $G_D$ if and only if $\lambda$ is even (respectively, if and only if $\lambda$ is a multiple of 3). It follows that $S''$ is unbreakably $k$-symmetric if and only if $\lambda$ is not a multiple of $k$. Since $S$ and $S''$ are equivalent (because both are equivalent to $S'$), Lemma~\ref{l:shapeequi} implies that $S$ is unbreakably $k$-symmetric if and only if its scale is not a multiple of $k$.
\end{proof}

\begin{figure}[ht]
\begin{center}
\includegraphics[width=1\linewidth]{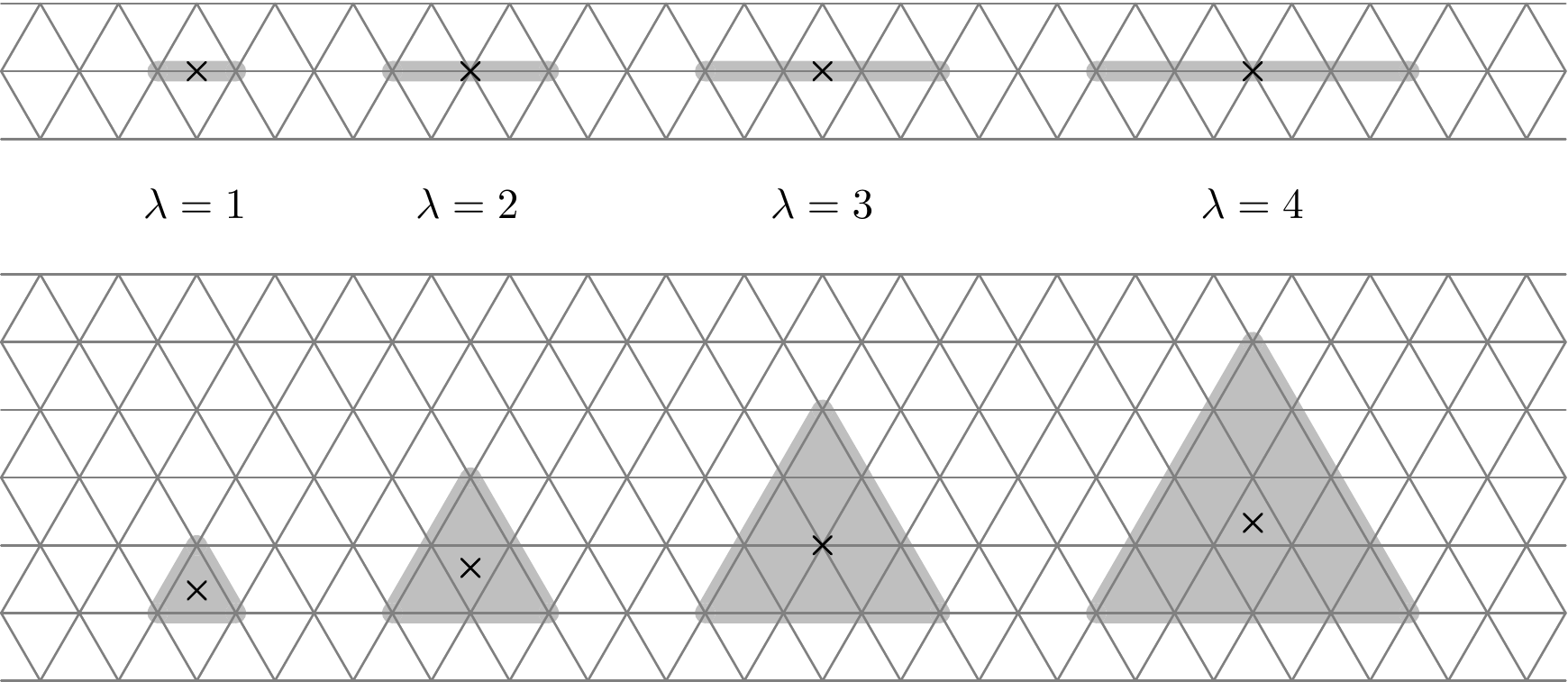}
\end{center}
\caption{If a shape is equivalent to an edge of $G_D$, its center is a vertex of $G_D$ if and only if its scale $\lambda$ is even. If a shape is equivalent to a face of $G_D$, its center is a vertex of $G_D$ if and only if its scale $\lambda$ is a multiple of $3$.}
\label{f:unbreakable}
\end{figure}

The term ``unbreakably'' is justified by the following theorem.

\begin{theorem}\label{t:neg}
If there exists an $(S_0,S_F)$-shape formation algorithm, and $S_0$ is unbreakably $k$-symmetric, then any minimal shape that is equivalent to $S_F$ is also unbreakably $k$-symmetric.
\end{theorem}
\begin{proof}
Assume that there is a $k$-fold rotation $\rho$ that leaves $S_0$ unchanged, and assume that its center is not a vertex of $G_D$. Then, the orbit of any vertex of $G_D$ under $\rho$ has period $k$. The system is naturally partitioned into \emph{symmetry classes} of size $k$: for any particle $p$ occupying a vertex $v$ of $G_D$ at stage 0, the symmetry class of $p$ is defined as the set of $k$ distinct particles that occupy the vertices $v$, $\rho(v)$, $\rho(\rho(v))$, \dots, $\rho^{k-1}(v)$.

Assume that the port labels of the $k$ particles in a same symmetry class are arranged symmetrically with respect to the center of $\rho$. Suppose that the system executes an $(S_0,S_F)$-shape formation algorithm and, at each stage, the scheduler picks a single symmetry class and activates all its members. It is easy to prove by induction that, at every stage, all active particles have equal views, receive equal messages from equally labeled ports, send equal messages to symmetric ports, and perform symmetric contraction and expansion operations. Note that, since only one symmetry class is active at a time, and the center of $\rho$ is not a vertex of $G_D$, no two particles ever try to expand toward the same vertex, and so no conflicts have to be resolved. Therefore, as the configuration evolves, it preserves its center of symmetry, and the system always forms an unbreakably $k$-symmetric shape. Since eventually the system must form a shape $S'_F$ equivalent to $S_F$, it follows that $S'_F$ is unbreakably $k$-symmetric. By Lemma~\ref{l:shapemin}, any minimal shape that is equivalent to $S'_F$ (or, which is the same, to $S_F$) is unbreakably $k$-symmetric, as well.
\end{proof}

In Section~\ref{s:3}, we are going to prove that the condition of Theorem~\ref{t:neg} characterizes the feasible pairs of shapes, provided that $S_0$ is simply connected, and the size of $S_0$ is large enough with respect to the base size of $S_F$.

\smallskip
\noindent\textbf{Measuring movements and rounds.}
We will also be concerned to measure the total number of \emph{moves} performed by a system of size $n$ executing a universal shape formation algorithm. This number is defined as the maximum, taken over all the feasible pairs $(S_0,S_F)$ where the size of $S_0$ is exactly $n$, all possible port labels, and all possible schedules, of the total number of contraction and expansion operations that all the particles collectively perform through the entire execution of the algorithm.

Next we will show a lower bound on the total number of moves of a universal shape formation algorithm.

\begin{theorem}\label{t:lower}
A system of $n$ particles executing any universal shape formation algorithm performs $\Omega(n^2)$ moves in total.
\end{theorem}
\begin{proof}
Let $d>0$ be an integer, and suppose that a system of $n=3d(d+1)+1$ particles forms a regular hexagon $H_d$ of side length $d$ at stage 0. Let the final shape $S_F$ be a single edge of $G_D$. We will show that any $(H_d,S_F)$-shape formation algorithm $A$ requires $\Omega(n^2)$ moves in total.

Since $S_F$ is an edge of $G_D$, a system executing $A$ from $H_d$ will eventually form a line segment $S'_F$ of length at least $n-1$ (recall that the particles do not have to be contracted to form a shape), say at stage $s$. Without loss of generality, we may assume that the center of $H_d$ lies at the origin of the Cartesian plane, $S'_F$ is parallel to the $x$ axis, and at stage $s$ the majority of particle's heads have non-negative $x$ coordinate.

Observe that, at stage 0, each particle's head has $x$ coordinate at most $d$. So, if the $x$ coordinate of a particle's head at stage $s$ is $k/2$, for some non-negative integer $k$, then the particle must have performed at least $\lceil k/2\rceil-d$ expansion operations between stage 0 and stage $s$. Also, if the particle does not occupy an endpoint of $S'_F$, there must be another particle whose head has $x$ coordinate at least $k/2+1$, which has performed at least $\lceil k/2+1\rceil-d$ expansion operations, etc.

It follows that a lower bound on the total number of moves that the system performs before forming $S'_F$ is
$$\sum_{i=0}^{\lceil n/2\rceil} (i-d)=\frac{\left\lceil \frac n2\right\rceil\left(\left\lceil \frac n2\right\rceil+1\right)}{2}-d\left(\left\lceil \frac n2\right\rceil+1\right)\geq \frac n2\left(\frac n4-d\right)=\Omega(n^2),$$
because $d=\Theta(\sqrt n)$.

To complete the proof, we should also show that $(H_d,S_F)$ is a feasible pair: indeed, the above lower bound would be irrelevant if there were no actual algorithm to form $S_F$ from $H_d$. However, we omit the tedious details of such an algorithm because we are going to prove a much stronger statement in Section~\ref{s:3}, where we characterize the feasible pairs in terms of their unbreakable $k$-symmetry. Note that the center of $H_d$ lies in a vertex of $G_D$, and hence $H_d$ is not unbreakably $k$-symmetric. Therefore, according to Theorem~\ref{tfinal}, $(H_d,S_F)$ is a feasible pair.
\end{proof}

In Section~\ref{s:3}, we will prove that our universal shape formation algorithm requires $O(n^2)$ moves in total, and is therefore asymptotically optimal with respect to this parameter.

Similarly, we want to measure how many \emph{rounds} it takes the system to form the final shape (a round is a span of time in which each particle is activated at least once). We will show that our universal shape formation algorithm takes $O(n^2)$ rounds.

\section{Universal Shape Formation Algorithm}\label{s:3}
\smallskip
\noindent\textbf{Algorithm structure.}
The universal shape formation algorithm takes a ``final shape'' $S_F$ as a parameter: this is encoded in the initial states of all particles. Without loss of generality, we will assume $S_F$ to be minimal. The algorithm consists of seven phases:
\begin{enumerate}
\item A \emph{lattice consumption phase}, in which the initial shape $S_0$ is ``eroded'' until 1, 2, or 3 pairwise adjacent particles are identified as ``candidate leaders''. No particle moves in this phase: only messages are exchanged. This phase ends in $O(n)$ rounds.
\item A \emph{spanning forest construction phase}, in which a spanning forest of $S_0$ is constructed, where each candidate leader is the root of a tree. No particle moves, and the phase ends in $O(n)$ rounds.
\item A \emph{handedness agreement phase}, in which all particles assume the same handedness as the candidate leaders (some candidate leaders may be eliminated in the process). In this phase, at most $O(n)$ moves are made. However, at the end, the system forms $S_0$ again. This phase ends in $O(n)$ rounds.
\item A \emph{leader election phase}, in which the candidate leaders attempt to break symmetries and elect a unique leader. If they fail to do so, and $k>1$ leaders are left at the end of this phase, it means that $S_0$ is unbreakably $k$-symmetric, and therefore the ``final shape'' $S_F$ must also be unbreakably $k$-symmetric (cf.~Theorem~\ref{t:neg}). No particle moves, and the phase ends in $O(n^2)$ rounds.
\item A \emph{straightening phase}, in which each leader coordinates a group of particles in the formation of a straight line. The $k$ resulting lines have the same length. At most $O(n^2)$ moves are made, and the phase ends in $O(n^2)$ rounds.
\item A \emph{role assignment phase}, in which the particles determine the scale of the shape $S'_F$ (equivalent to $S_F$) that they are actually going to form. Each particle is assigned an identifier that will determine its behavior during the formation process. No particle moves, and the phase ends in $O(n^2)$ rounds.
\item A \emph{shape composition phase}, in which each straight line of particles, guided by a leader, is reconfigured to form an equal portion of $S'_F$. At most $O(n^2)$ moves are made, and the phase ends in $O(n^2)$ rounds.
\end{enumerate}
No a-priori knowledge of $S_0$ is needed to execute this algorithm ($S_0$ just has to be simply connected), while $S_F$ must of course be known to the particles and have constant size, so that its description can reside in their memory. Note that the knowledge of $S_F$ is needed only in the last two phases of the algorithm.

\smallskip
\noindent\textbf{Synchronization.}
As long as there is a unique (candidate) leader $p$ in the system, there are no synchronization problems: $p$ coordinates all other particles, and autonomously decides when each phase ends and the next phase starts.

However, if there are $k>1$ (candidate) leaders, there are possible issues arising from the intrinsic asynchronicity of our particle model. Typically, a (candidate) leader will be in charge of coordinating only a portion of the system, and we want to avoid the undesirable situation in which different leaders are executing different phases of the algorithm.

To implement a basic synchronization protocol, we will ensure three things:
\begin{itemize}
\item All the (candidate) leaders must always be pairwise adjacent, except perhaps in the last two phases of the algorithm (i.e., the role assignment phase and the shape composition phase) and for a few stages during the handedness agreement phase and the straightening phase.
\item Every time a (candidate) leader is activated, it sends all other (candidate) leaders a message containing an identifier of the phase that it is currently executing (the message may also contain other data, depending on the phase of the algorithm).
\item Whenever a (candidate) leader transitions from a phase into the next, it waits for all other (candidate) leaders to be in the same phase (except when it transitions to the last phase).
\end{itemize}

This basic protocol is executed ``in parallel'' with the main algorithm, and it always works in the same way in every phase. In the following, we will no longer mention it explicitly, but we will focus on the distinctive aspects of each phase.

\subsection{Lattice Consumption Phase}\label{s:3.1}
\smallskip
\noindent\textbf{Algorithm.}
The goal of this phase is to identify 1, 2, or 3 \emph{candidate leaders}. This is done without making any movements, but only exchanging messages. Each particle's internal state has a flag (i.e., a bit) called \emph{Eligible}. All particles start the execution in the same state, with the Eligible flag set. As the execution proceeds, Eligible particles will gradually eliminate themselves by clearing their Eligible flag. This is achieved through a process similar to erosion, which starts from the boundary of the initial shape and proceeds toward its interior.

Suppose that all the particles in the system are contracted (which is true at stage 0). Then, we define four types of \emph{corner particles}, which will start the erosion process:
\begin{itemize}
\item a \emph{0-corner particle} is an Eligible particle with no Eligible neighbors;
\item a \emph{1-corner particle} is an Eligible particle with exactly one Eligible neighbor;
\item a \emph{2-corner particle} is an Eligible particle with exactly two Eligible neighbors $p_1$ and $p_2$, such that $p_1$ is adjacent to $p_2$;
\item a \emph{3-corner particle} is an Eligible particle with exactly three Eligible neighbors $p_1$, $p_2$, $p_3$, with $p_2$ adjacent to both $p_1$ and $p_3$. $p_2$ is called the \emph{middle neighbor}.
\end{itemize}
We say that a particle $p$ is \emph{locked} if it is a 3-corner particle and its middle neighbor $p'$ is also a 3-corner particle. If $p$ is locked, it follows that $p'$ is locked as well, with $p$ its middle neighbor. In this case, we say that $p$ and $p'$ are \emph{companions}.

\begin{figure}[ht]
\begin{center}
\includegraphics{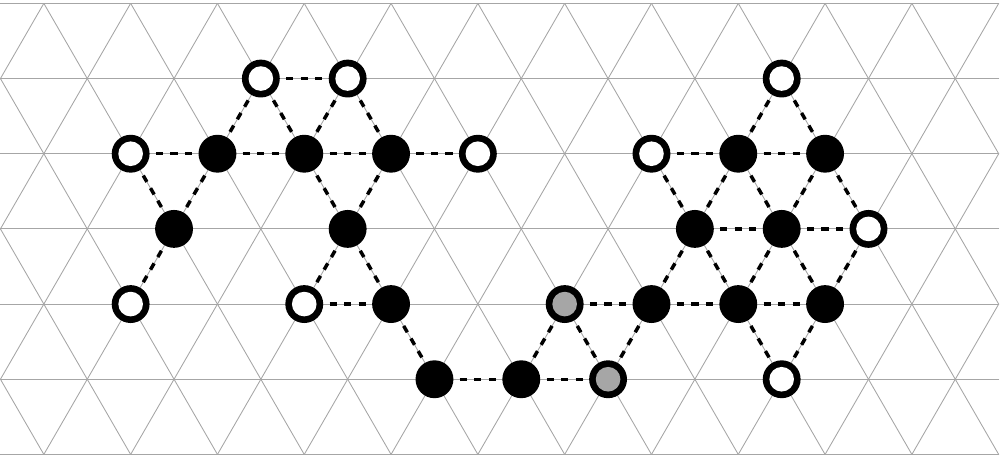}
\end{center}
\caption{The particles in white or gray are corner particles; the two in gray are locked companion particles. Dashed lines indicate adjacencies between particles.}
\label{f:consumption}
\end{figure}

Each particle has also another flag, called \emph{Candidate} (initially not set), which will be set if the particle becomes a candidate leader. Whenever a particle is activated, it sends messages to all its neighbors, communicating the state of its Eligible and Candidate flags. When a particle receives such a message, it memorizes the information it contains (perhaps overwriting outdated information). Therefore, each particle keeps an updated copy of the Eligible and Candidate flags of each of its neighbors. Once a particle knows the states of all its neighbors (i.e., when it has received messages from all of them), it also knows if it is a $k$-corner particle. If it is, it broadcasts the number $k$ to all its neighbors every time it is activated. In turn, the neighbors memorize this number and keep it updated.

There is a third flag, called \emph{Stable} (initially not set), which is cleared whenever a particle receives a message from a neighbor communicating that its internal state has changed. Otherwise, the flag is set, meaning that the states of all neighbors have been stable for at least one stage.

The following rules are also applied by active particles, alongside with the previous ones:
\begin{itemize}
\item If a particle $p$ does not know whether some if its neighbors were corner particles the last time they were activated (because it has not received enough information from them, yet), it waits.
\item Otherwise, $p$ knows which of its neighbors were corner particles the last time they were activated, and in particular which of them are Eligible. This also implies that $p$ knows if it is a corner particle, and if it is locked (for a proof, see Theorem~\ref{tp1}). If $p$ is not a corner particle or if it is locked, it waits.
\item Otherwise, $p$ is a non-locked corner particle. If its Eligible flag is not set, or if its Candidate flag is set, or if its Stable flag is not set, it waits.
\item Otherwise, $p$ changes its flags as follows.
\begin{itemize}
\item If $p$ is a 0-corner particle, it sets its own Candidate flag.
\item Let $p$ be a 1-corner particle. If its unique Eligible neighbor was a 1-corner particle the last time it was activated, $p$ sets its own Candidate flag; otherwise, $p$ clears its own Eligible flag.
\item Let $p$ be a 2-corner particle. If both its Eligible neighbors were 2-corner particles the last time they were activated, $p$ sets its own Candidate flag; otherwise, $p$ clears its own Eligible flag.
\item If $p$ is a 3-corner particle, it clears its own Eligible flag.
\end{itemize}
\end{itemize}

\smallskip
\noindent\textbf{Correctness.}
\begin{lemma}\label{lp1:corner}
If a system of contracted Eligible particles forms a simply connected shape, then there is a corner particle that is not locked.
\end{lemma}
\begin{proof}
Let $S$ be a simply connected shape formed by a system $P$ of contracted Eligible particles. For each pair of companion locked particles of $P$, let us remove one. At the end of this process, we obtain a reduced system $P'$ with no locked particles that again forms a simply connected shape $S'$. Note that all corner particles of $P'$ are non-locked corner particles of $P$; hence, it suffices to prove the existence of a corner particle in $P'$.

By definition of shape, $S'$ can be decomposed into maximal 2-dimensional polygons interconnected by 1-dimensional polygonal chains, perhaps with ramifications (if two polygons are connected by one vertex, we treat this vertex as a polygonal chain with no edges). Since $S'$ is simply connected, the abstract graph obtained by collapsing each maximal polygon of $S'$ into a single vertex forms a tree, which has at least one leaf. The leaf may represent the endpoint of a polygonal chain of $S'$, which is the location of a 1-corner particle of $P'$. Otherwise, the leaf represents a polygon $S''\subseteq S'$, perhaps connected to the rest of $S'$ by a polygonal chain with an endpoint in a vertex $v$ of $S''$. Since $S''$ is a polygon, it has at least three convex vertices, at most one of which is $v$. Each other convex vertices of $S''$ is therefore the location of a 2-corner particle or a 3-corner particle of $P'$.
\end{proof}

\begin{lemma}\label{lp1:connected}
If a system of contracted Eligible particles forms a simply connected shape, and any set of non-locked corner particles is removed at once, the new system forms again a simply connected shape.\footnote{For convenience, with a little abuse of terminology, we treat single vertices of $G_D$ and the empty set as shapes, even if technically they are not, according to the definitions of Section~\ref{s:2}.}
\end{lemma}
\begin{proof}
Instead of removing all the particles in the given set $C$ at once, we remove them one by one in any order, and use induction to prove our lemma. Let $P$ be a system of contracted Eligible particles forming a simply connected shape, and let $C'$ be a set of corner particles of $P$ such that, if a particle of $C'$ is locked, then its companion is not in $C'$. This condition is obviously satisfied by the given set $C$, because it does not contain locked corner particles at all.

Let $p\in C'$, and let $P'=P\setminus \{p\}$. We have to prove that $P'$ forms a simply connected shape, that $C'\setminus \{p\}$ is a set of corner particles of $P'$, and that $C'\setminus \{p\}$ does not contain any pair of companion locked particles of $P'$. Note that the fact that $P'$ forms a simply connected shape is evident, since $p$ is a corner particle of $P$. Therefore, any path in the subgraph of the grid $G$ induced by $P$ that goes through $p$ may be re-routed through the neighbors of $p$ in $P$.

Now, if $C'=\{p\}$, there is nothing else to prove. So, let $p'\neq p$ be another particle of $C'$. Suppose first that $p'$ is adjacent to $p$. Hence, removing $p$ reduces the number of Eligible neighbors of $p'$ by one. The only case in which $p'$ could cease to be a corner particle would be if $p$ were its middle neighbor, implying that $p$ and $p'$ would be locked companions. But this is not possible, because by assumption $C'$ does not contain a pair of locked companion particles. Therefore, $p'$ is necessarily a corner particle of $P'$. Note that $p'$ cannot be a 3-corner particle of $P'$, because one of its Eligible neighbors (namely, $p$) has been removed. Hence $p'$ cannot be locked in $P'$.

Suppose now that $p'$ is not adjacent to $p$. Then, obviously, $p'$ is a corner particle of $P'$, since removing $p$ does not change its neighborhood. We only have to prove that, if $p'$ is a locked 3-corner particle in $P'$, then $C'\setminus \{p\}$ does not contain the companion of $p'$. Assume the opposite: let $p'$ be a locked particle in $P'$, let $p''$ be its companion, and let $p'$ and $p''$ be in $C'\setminus \{p\}$, and hence in $C'$. By assumption, $p''$ cannot be locked in $P$, or else $C'$ would contain a pair of locked companion particles. So, $p$ must be adjacent to $p''$. Removing $p$ reduces the number of Eligible neighbors of $p''$, implying that $p''$ must have four Eligible neighbors in $P$ (since $p''$ is a 3-corner particle in $P'$). It follows that $p''$ cannot be a corner particle of $P$, and therefore it cannot be in $C'$, contradicting our assumption.
\end{proof}

\begin{theorem}\label{tp1}
Let $P$ be a system of $n$ contracted Eligible particles forming a simply connected shape $S_0$ at stage 0. If all particles of $P$ execute the lattice consumption phase of the algorithm, there is a stage $s$, reached in $O(n)$ rounds, where there are 1, 2, or 3 pairwise adjacent Candidate particles, and all other particles are non-Eligible. Moreover, at all stages from 0 to $s$, the system forms $S_0$, and the sub-system of Eligible particles forms a simply connected shape.
\end{theorem}
\begin{proof}
Recall that, in the lattice consumption phase, a particle never changes its Eligible or Candidate flags unless it is an Eligible, non-Candidate, Stable, non-locked corner particle that has enough information about its neighbors.

Whenever a particle $p$ is activated and reads the pending messages, everything it knows about the internal flags of its neighbors is correct and up to date. Indeed, these flags can be changed only by the neighbors themselves when they are activated, and whenever this happens they send the updated values to $p$. Therefore, $p$ always reads the most recent values of the flags of its neighbors, no matter how and when the scheduler activates them. So, $p$ is able to correctly determine if it is a corner particle by just looking at the Eligible flags of its neighbors and how they are arranged.

On the other hand, when $p$ receives a message from a neighbor $p'$ claiming that $p'$ is or is not a $k$-corner particle, this information may be outdated, because a neighbor of $p'$ may have eliminated itself, and $p'$ may have been inactive ever since. However, $p$ is still able to determine if it is locked or not. Indeed, suppose that $p$ has correctly determined that it is a 3-corner particle, implying that it currently has three consecutive Eligible neighbors $p_1$, $p_2$, and $p_3$. Suppose that its middle neighbor $p_2$ has claimed to be a 3-corner particle in its last message. This statement was correct the last time $p_2$ was active, implying that $p_2$ had three consecutive Eligible neighbors. Because particles can become non-Eligible but never become Eligible again, the three Eligible particles that $p_2$ saw must be $p_1$, $p$, and $p_2$, since they are currently Eligible. It follows that $p_2$ is still a 3-corner particle and hence $p$ is locked. The converse is also true, for the same reason.

We deduce that a particle will eliminate itself only if it truly is a non-locked corner particle. Also note that no particle is allowed to move during the lattice consumption phase. So, the system will always form the same shape $S_0$, and, by Lemma~\ref{lp1:connected}, the sub-system of Eligible particles will always be simply connected.

Next we prove that, if a particle ever sets its Candidate flag, then there is a stage where there are 1, 2, or 3 pairwise adjacent Candidates, and all other particles are non-Eligible. Say that  at some point a particle $p$ becomes a Candidate, which means that it was able to determine that it is a $k$-corner particle, with $0\leq k\leq 2$.

If $k=0$, then $p$ is the only Eligible particle left, because the sub-system of Eligible particles must be connected. If $k=1$, then $p$ has a unique Eligible neighbor $p'$, which was a 1-corner particle the last time it was activated. This means that the only Eligible neighbor of $p'$ was $p$, and hence $p$ and $p'$ are the only Eligible particles in the system. Eventually, $p'$ will become Stable and will either eliminate itself or become a Candidate. If $k=2$, then $p$ has two adjacent Eligible neighbors $p'$ and $p''$, which were 2-corner particles the last time they were activated. So, the only Eligible particles in the system are $p$, $p'$, and $p''$, which are pairwise adjacent. Both $p'$ and $p''$ will eventually become Stable, and they will either eliminate themselves or become Candidates.

We now have to prove that Eligible particles steadily eliminate themselves until only Candidates are left. Assume the opposite, and suppose that the execution of the algorithm reaches a point where Eligible particles stop becoming non-Eligible. By the above reasoning, we may assume that the system contains no Candidate particles at this point. As the sub-system of Eligible particles is simply connected at any time, by Lemma~\ref{lp1:corner} there are non-Candidate non-locked corner particles. Since no particle ever changes its internal flags again, all of them will eventually become Stable. So, there will be an Eligible, non-Candidate, Stable, non-locked corner particle that will either become non-Eligible or a Candidate, which contradicts our assumptions.

It remains to prove that at least one particle will become a Candidate. Assume the opposite. At each stage, some non-locked corner particles possibly eliminate themselves, and this process goes on until there are no Eligible particles left. Let $s$ be the stage when the last Eligible particles eliminate themselves (simultaneously). As all of them have to be non-locked corner particles at stage $s$, it is easy to see that only three configurations are possible:
\begin{itemize}
\item At stage $s$ there is only one Eligible particle. Since this is a 0-corner particle, according to the algorithm it will become a Candidate.
\item At stage $s$ there are only two adjacent 1-corner particles $p$ and $p'$. Recall that a particle has to be Stable in order to eliminate itself. Since $p$ is Stable at stage $s$, it means that there is a stage $s'<s$ during which $p$ has already sent a message to $p'$ saying that it was a 1-corner particle (otherwise, some neighbor of $p$ would have eliminated itself in the meantime, implying that $p$ would not be Stable). Since $p'$ is active at stage $s$, it must receive or have already received the message sent at time $s'$ by $p$. So, $p'$ knows that $p$ is a 1-corner, and hence it becomes a Candidate (and vice versa).
\item At stage $s$ there are only three pairwise adjacent 2-corner particles. Reasoning as in the previous case, we see that, since all three particles are Stable at stage $s$, they know that they are all 2-corner particles, and therefore they become Candidates.
\end{itemize}
In all cases, at least one particle becomes a Candidate, contradicting our assumption.

Finally, the upper bound of $O(n)$ rounds easily follows from the fact that, in a constant number of rounds, at least one corner particle becomes non-Eligible or a Candidate. This happens at most $n-1$ times, until only Candidates are left.
\end{proof}

\subsection{Spanning Forest Construction Phase}\label{s:3.2}
\smallskip
\noindent\textbf{Algorithm.}
The spanning forest construction phase starts when 1, 2, or 3 pairwise adjacent candidate leaders have been identified, and no other particle is Eligible. In this phase, each candidate leader becomes the root of a tree embedded in $G$. Eventually, the set of these trees will be a spanning forest of the subgraph of $G$ induced by the system $P$.

Each particle has a flag called \emph{Tree}, initially not set, whose purpose is to indicate that the particle has been included in a tree. Moreover, each particle also has a variable called \emph{Parent}, which contains the local port number corresponding to its parent, provided that the particle is part of a tree (the initial value of this variable is $-1$).

As in the previous phase, all particles send information to their neighbors containing part of their internal states. This information is recorded by the receiving particles: so, each particle has an \emph{Is-in-Tree} flag and an \emph{Is-my-Child} flag corresponding to each neighbor. All these flag are initially not set.

Finally, there is a \emph{Tree-Done} flag (initially not set) corresponding to each neighbor, which is used in the last part of the phase.

The following rules apply to all particles during the spanning forest construction phase:
\begin{itemize}
\item If a particle's Candidate flag is set and its Tree flag is not set, it sets its own Tree flag and leaves its own Parent flag to $-1$ (implying that it is the root of a tree).
\item If a particle's Tree flag is set, it sends a \emph{Parent} message to the port corresponding to its own Parent variable (assuming it is not $-1$), and it sends a \emph{Tree} message to all other neighbors.
\item If a particle receives a Tree message from a neighbor, it sets the Is-in-Tree flag relative to its port. Similarly, if it receives a Parent message from a neighbor, it sets both the Is-in-Tree and the Is-my-Child flag relative to its port.
\item If a particle's Candidate flag is not set, its Tree flag is not set, and the Is-in-Tree flags relative to some of its neighbors are set, then it sets its own Tree flag. Let $k$ be the smallest port number corresponding to a neighbor whose relative Is-in-Tree flag is set. Then, the particle sets its own Parent flag to $k$ (implying that that neighbor is now its parent).
\item If a particle $p$'s Tree flag is set, the Is-in-Tree flags corresponding to all its neighbors are set, and the Tree-Done flags relative to all it children are set (recall that its children are the neighbors whose relative Is-my-Child flag is set), then:
\begin{itemize}
\item If $p$ has a parent (i.e., its Parent variable is not $-1$), it sends a \emph{Tree-Done} message to its parent.
\item If $p$ has no parent (i.e., it is a candidate leader), it sends a Tree-Done message to its Candidate neighbors.
\end{itemize}
\item If a particle receives a Tree-Done message from one of its children, it sets the corresponding Tree-Done flag.
\end{itemize}

\smallskip
\noindent\textbf{Correctness.}
\begin{theorem}\label{tp2}
Let $P$ be the system resulting from Theorem~\ref{tp1}. If all particles of $P$ execute the spanning forest construction phase of the algorithm, then there is a stage, reached after $O(n)$ rounds, where every particle has the Tree flag set, each non-Candidate particle has a unique parent, and each particle has received a Tree-Done message from all its children. No particle moves in this phase.
\end{theorem}
\begin{proof}
It is easy to prove by induction that, at every stage, the Tree particles form a forest with a tree rooted in each Candidate particle, and that the Parent variables of all Tree particles are consistent. Indeed, when any (positive) number of neighbors of a non-Tree particle $p$ become Tree particles and start sending Tree messages to $p$, $p$ chooses one of them as its parent as soon as it is activated, and sets its flags accordingly. It then communicates this change to its neighbors, which update their Is-in-Tree and the Is-my-Child flags consistently.

Since $P$ forms a connected shape (because it results from Theorem~\ref{tp1}), eventually all particles become part of some tree, and a spanning forest is constructed. So, eventually, some leaves of the forest observe that all their neighbors are Tree particles (because all their relative Is-in-Tree flags are set) and none of them is their child (because none of their relative Is-my-Child flags is set). These leaves send Tree-Done messages to their parents.

As more leaves send Tree-Done messages to their parents, some internal particles start observing that all their children are sending Tree-Done messages, and all other neighbors are Tree particles. These internal particles therefore send Tree-Done messages to their parents, as well. Eventually, the Candidate particles will receive Tree-Done messages from all their children. At this stage, every particle has the Tree flag set, each non-Candidate particle has a unique parent, and each particle has received a Tree-Done message from all its children.

To show that the phase ends in $O(n)$ rounds, it suffices to note that in a constant number of rounds either a new particle sets its Tree flag or forwards the Tree-Done message to its parent (or its Candidate neighbors if it has no parent).
\end{proof}

\subsection{Handedness Agreement Phase}\label{s:3.3}
When a candidate leader has received Tree messages from all its neighbors and Tree-Done messages from all its children, it transitions to the handedness agreement phase. Recall that each particle may label ports in clockwise or counterclockwise order: this is called the particle's handedness. By the end of this phase, all particles will agree on a common handedness. The agreement process starts at the candidate leaders and proceeds through the spanning forest constructed in the previous phase, from parents to children.

\smallskip
\noindent\textbf{Agreement among candidate leaders.}
In the first stages of this phase, the candidate leaders agree on a common handedness. This may result in the ``elimination'' of some of them. If there is a unique candidate leader, this part of the algorithm is trivial. So, let us assume that there are two or three candidate leaders.

Suppose that there are two candidate leaders $p$ and $p'$. Then, they have exactly two neighboring vertices $u$ and $v$ in common. By now, the candidate leaders know if $u$ and $v$ are occupied or not. There are three cases.
\begin{itemize}
\item Exactly one between $u$ and $v$ is occupied. Without loss of generality, $u$ is occupied by a particle $p_u$, and $v$ is unoccupied. Then, both $p$ and $p'$ send a \emph{You-Choose} message to $p_u$. When $p_u$ has received You-Choose messages from both, it arbitrarily picks one between $p$ and $p'$, say $p$. Then $p_u$ sends a \emph{Chosen} message to $p$ and a \emph{Not-Chosen} message to $p'$. As a consequence, $p'$ ceases to be a candidate leader (by clearing its own Candidate and Eligible flags), and $p$ becomes the parent of $p'$ (i.e., the Parent variable of $p'$ and the Is-my-Child variables of $p$ are appropriately updated).
\item $u$ is occupied by a particle $p_u$ and $v$ is occupied by a particle $p_v$. Let the edge $\{p,p'\}$ be labeled $i$ by $p$, and observe that $i-\ell(p,u)\equiv -i+\ell(p,v)\equiv \pm 1 \pmod 6$. Without loss of generality, $i-\ell(p,u)\equiv 1 \pmod 6$. Then, $p$ sends a \emph{You-Choose} message to $p_u$ and a \emph{You-do-not-Choose} message to $p_v$. $p'$ does the same. If one between $p_u$ or $p_v$ receives both You-Choose messages, it arbitrarily eliminates one between $p$ and $p'$, as explained above. Otherwise, both $p_u$ and $p_v$ receive a You-Choose message and a You-do-not-Choose message. This means that $p$ and $p'$ have the same handedness. So, $p_u$ and $p_v$ send \emph{Same-Handedness} messages to both $p$ and $p'$, who wait until they receive both messages.
\item Both $u$ and $v$ are unoccupied. As above, if the edge $\{p,p'\}$ is labeled $i$ by $p$, then $i-\ell(p,u)\equiv -i+\ell(p,v)\equiv \pm 1 \pmod 6$. Without loss of generality, assume that $i-\ell(p,u)\equiv 1 \pmod 6$. Then, $p$ attempts to expand toward $u$. Meanwhile, $p'$ does the same.
\begin{itemize}
\item If $p$ fails to expand toward $u$, it means that $p'$ has done it ($p$ realizes that this has happened because it cannot see its own tail the next time it is activated). In this case, $p$ sends an \emph{I-am-Eliminated} message to $p'$.
\item Suppose now that $p$ manages to expand toward $u$ (it realizes that it has expanded because it sees its own tail the next time it is activated). Then, $p$ looks back at $p'$, which is found at port $(i+1)\mod 6$ (see Figure~\ref{f:handedness1}). If $p$ sees a tail or an unoccupied vertex, it understands that $p'$ has expanded toward $v$. In this case, $p$ and $p'$ have the same handedness, and $p$ memorizes this information. If $p$ sees the head of $p'$, it sends a \emph{You-are-Eliminated} message to $p'$.
\end{itemize}
After this, if $p$ is still expanded, it contracts into $u$, it expands toward its original vertex, and contracts again. $p'$ does the same. Eventually, at least one between $p$ and $p'$ has realized that their handedness is the same, or has received a You-are-Eliminated or an I-am-Eliminated message. This information is shared by $p$ and $p'$ again when they are both in their initial positions and contracted. If one of them has to be eliminated, it does so by clearing its Candidate and Eligible flags, and becomes a child of the other candidate leader, as explained above.
\end{itemize}

\begin{figure}[ht]
\begin{center}
\includegraphics[scale=1]{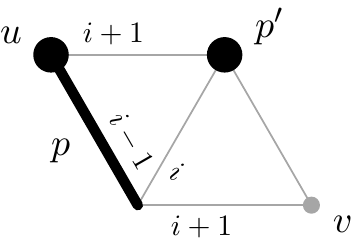}
\end{center}
\caption{The case of the agreement protocol between candidate leaders in which $u$ and $v$ are unoccupied. If $p$ expands toward the vertex corresponding to the label $(i-1)\mod 6$, it finds $p'$ at the vertex corresponding to the label $(i+1)\mod 6$.}
\label{f:handedness1}
\end{figure}

Suppose now that there are three candidate leaders $p$, $p'$, and $p''$. So, $p$ knows that there are candidate leaders corresponding to its local ports $i$ and $i'$, with $i'\equiv i+1 \pmod 6$. Then, $p$, sends an \emph{I-Choose-You} message through port $i$ and an \emph{I-do-not-Choose-You} message through port $i'$. Meanwhile, $p'$ and $p''$ do the same. Eventually, each candidate leader receives two messages.
\begin{itemize}
\item If one of them, say $p$, receives two I-Choose-You messages, it sends You-are-Eliminated messages to both $p'$ and $p''$. Then, $p'$ and $p''$ cease to be candidate leaders and become children of $p$.
\item Otherwise, each candidate leaders receives one I-Choose-You and one I-do-not-Choose-You message. This means that all candidate leaders have the same handedness. Each of them sends an \emph{I-am-not-Chosen} message to the others. When they receive each other's messages, they realize that their handedness is the same.
\end{itemize}

\smallskip
\noindent\textbf{Basic handedness communication.}
As a basic operation, we want to let a parent ``impose'' its own handedness onto a child. Of course, this cannot be done by direct communication, and we will therefore need a special \emph{handedness communication technique}, which we describe next.

Assume that a contracted particle $p$ intends to communicate its handedness to one of its children, a contracted particle $p'$. For now, we will make the simplifying assumption that all other particles are contracted and idle. We will show later how to handle the general case in which several particles are operating in parallel.

Say that the edge $\{p,p'\}$ is labeled $i$ by $p$ and $i'$ by $p'$. There are exactly two vertices $u$ and $v$ of $G_D$ that are adjacent to both $p$ and $p'$. Suppose first that at least one between $u$ and $v$ is not occupied by any particle. If both are unoccupied, $p$ will arbitrarily choose one of them. Without loss of generality, let us assume that $u$ is unoccupied, and $p$ has chosen it. Let $j=(\ell(p,u)-i)\mod 6$, and observe that $j=\pm 1$, since $u$ is adjacent to $p'$.
\begin{itemize}
\item $p$ memorizes $i$ and $j$, and expands toward $u$.
\item Then, $p$ computes the port corresponding to $p'$ as $(i-j)\mod 6$, and sends $p'$ a \emph{Handedness-A} message containing $j$.
\item Say $p'$ receives the Handedness-A message from port $i''$, and let $j'=(\mbox{Parent}-i'')\mod 6$ (recall from Section~\ref{s:3.2} that $\mbox{Parent}=i'$, because $p$ is the parent of $p'$). Now, $p$ and $p'$ have the same handedness if and only if $j=j'$. So, $p'$ memorizes this information and replies with a \emph{Handedness-OK} to port $i''$.
\item When $p$ receives the Handedness-OK message, it contracts into $u$.
\item Then, $p$ expands toward its original location and contracts again.
\end{itemize}

Suppose now that $u$ and $v$ are both occupied by particles $p_u$ and $p_v$, respectively: this case is illustrated in Figure~\ref{f:handedness2}. We say that $p_u$ and $p_v$ are \emph{auxiliary particles}.
\begin{itemize}
\item $p$ sends a \emph{Lock} message to both $p_u$ and $p_v$ (the purpose of this message will be explained later).
\item $p_u$ and $p_v$ reply by sending \emph{Locked} messages back to $p$.
\item When $p$ has received Locked messages from both $p_u$ and $p_v$, it sends a \emph{Get-Ready} message to $p'$.
\item When $p'$ receives the Get-Ready message, it sets an internal \emph{Ready} flag and sends an \emph{I-am-Ready} message back to $p$ (the purpose of the Ready flag will be explained later).
\item When $p$ receives the I-am-Ready message from $p'$, it sends the number $(\ell(p,u)-i)\mod 6$ to $p_u$ and the number $(\ell(p,v)-i)\mod 6$ to $p_v$.
\item Say that $p_u$ receives the number $j$ from $p$. Then, $p_u$ sends a \emph{Handedness-B} message containing the number $j$ to the (at most two) common neighbors of $p$ and $p_u$. Note that $p'$ is one of these neighbors (see Figure~\ref{f:handedness2}). $p_v$ does the same thing.
\item Whenever a particle receives a Handedness-B message from a neighbor, it responds with a \emph{Handedness-B-Acknowledged} to the same neighbor.
\item Say that $p'$ receives a Handedness-B message containing the number $j=(\ell(p,u)-i)\mod 6$ from $p_u$, and say that $\ell(p',u)=i''$. As before, $p'$ computes $j'=(\mbox{Parent}-i'')\mod 6$, and determines if it has the same handedness as $p$ by comparing $j$ and $j'$. If $p'$ receives a number from $p_v$, it does the same thing.
\item When $p'$ has received numbers from both $p_u$ and $p_v$, it sends a Handedness-OK message to $p$.
\item When $p$ receives the Handedness-OK message from $p'$, it sends \emph{Unlock} messages to both $p_u$ and $p_v$.
\item When $p_u$ and $p_v$ receive an Unlock message from $p$ and a Handedness-B-Acknowledged message from every neighbor to which they sent Handedness-B messages, they send an \emph{Unlocked} message back to $p$.
\item $p$ waits until it receives Unlocked messages from both $p_u$ and $p_v$.
\end{itemize}

\begin{figure}[ht]
\begin{center}
\includegraphics[scale=1]{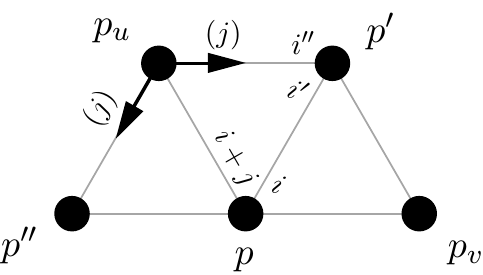}
\end{center}
\caption{The case of the handedness communication protocol in which $u$ and $v$ are occupied. Arrows indicate messages. The particles $p$ and $p'$ have the same handedness if and only if $j\equiv i'-i''\pmod 6$.}
\label{f:handedness2}
\end{figure}

\smallskip
\noindent\textbf{Main handedness agreement algorithm.}
The main part of the handedness agreement algorithm starts when all the candidate leaders have the same handedness. Next we describe the main algorithm for a generic particle $p$.

If $p$ is a candidate leader or if it receives a \emph{Begin-Handedness-Communication} message from its parent, $p$ starts communicating its handedness to its children. $p$ picks one child and executes the handedness communication technique with it. Then it does so with the next child, etc.

When a child $p'$ realizes that its handedness is not the handedness of $p$, it sets a special internal flag that reminds it to apply the function $f(i)=5-i$ to all its port labels. If the flag is not set, $f$ is the identity function. The composition $\widetilde\ell=f\circ\ell$, where $\ell$ is the labeling of $p'$, is called the \emph{corrected labeling} of $p'$, and will be used by $p'$ instead of $\ell$. In other terms, $p'$ ``pretends'' to have the handedness of $p$, and it behaves accordingly for the rest of the execution of the shape formation algorithm.

When $p$ has communicated its handedness to all its children, it sends a Begin-Handedness-Communication message to its first child. Then $p$ waits until the child has sent it a \emph{Done-Handedness-Communication} message back. Subsequently, $p$ sends a Begin-Handedness-Communication message to its second child, etc.

When the last child of $p$ has sent a Done-Handedness-Communication message to it (or if $p$ is a leaf of the spanning forest of $P$), $p$ sends a Done-Handedness-Communication message back to its father (provided that $p$ is not a candidate leader).

\smallskip
\noindent\textbf{Resolving conflicts.}
Note that several pairs of particles may be executing the handedness communication technique at the same time: precisely, as many as the trees in the spanning forest, i.e., as many as the candidate leaders. These particles may interfere with each other when they try to expand toward the same vertex or when they send messages to the same particle. In the following, we explain how these conflicts are resolved.

To begin with, each particle $p$ memorizes which of its surrounding vertices are initially occupied by other particles. Then, when $p$ executes the handedness communication technique, it looks at its neighboring vertices $u$ and $v$. If any of them is supposed to be occupied but is currently unoccupied or it is a tail vertex, $p$ waits.

Similarly, if $p$ fails to expand toward a supposedly empty vertex $u$ because another particle has expanded toward it at the same time, $p$ waits until $u$ is unoccupied again (recall that $p$ realizes that its expansion attempt has failed if it cannot see its own tail).

After an auxiliary particle $p_u$ has sent a Locked message to $p$, it ignores all Lock messages from any other particle until it has sent an Unlocked message back to $p$. This prevents $p_u$ from becoming an auxiliary particle in two independent handedness communication operations simultaneously.

Similarly, when $p_u$ is an auxiliary particle of $p$ and $p'$, it sends Handedness-B messages to its common neighbors with $p$. So, another particle $p''\neq p'$, which is not involved in the operation, might receive this message and behave incorrectly. Three situations are possible:
\begin{itemize}
\item If $p''$ has already been the recipient of a (completed) handedness communication operation, it simply ignores this message.
\item If $p''$ has never been the recipient of a handedness communication operation, its Ready flag is still not set. So, $p''$ just responds to $p_u$ with a Handedness-B-Acknowledged message without doing anything. On the other hand, $p_u$ will not become unlocked until it has received this message. Therefore, when $p''$ will indeed be involved in a handedness communication operation, there will not be a pending Handedness-B message directed to $p''$.
\item Suppose that $p''$ is currently involved in a handedness communication operation. We claim that $p_u$ cannot be an auxiliary particle of this operation. This is because $p_u$ has already been locked by $p$, who is involved in an operation with $p'$, and therefore it cannot be an auxiliary particle in any other operation. Therefore, when $p''$ receives a Handedness-B message from $p_u$, it ignores it because it knows that $p_u$ is not its auxiliary particle ($p''$ only responds with a Handedness-B-Acknowledged message).
\end{itemize}

\smallskip
\noindent\textbf{Correctness.}
\begin{theorem}\label{tp3}
Let $P$ be the system resulting from Theorem~\ref{tp2}, forming a shape $S_0$. If all particles of $P$ execute the handedness agreement phase of the algorithm, then there is a stage, reached after $O(n)$ rounds, where all candidate leaders have received a Done-Handedness-Communication message from all their children. At this stage, $P$ forms $S_0$ again, all candidate leaders have the same handedness, and each other particle knows whether it has the same handedness as the candidate leaders. In this phase, at most $O(n)$ moves are performed in total.
\end{theorem}
\begin{proof}
The agreement protocol among candidate leaders works in a straightforward way in every case. Indeed, only the candidate leaders are ever allowed to move, and the other particles never send any message unless prompted by the candidate leaders themselves.

Eventually, all candidate leaders have the same handedness, and the main part of the handedness agreement phase starts. We have already proved that there can be no conflicts, in that particles involved in different handedness communication operations do not interfere with one another. We only have to prove that there can be no \emph{deadlocks}, and hence the execution never gets stuck. There are essentially three ways in which a deadlock may occur, which will be examined next.

The first potential deadlock situation is the one in which a particle $p$ intends to expand toward a vertex $u$ that was originally unoccupied, but now is occupied by some other particle $q$. According to the protocol, $p$ has to wait for $u$ to be unoccupied again. However, while $p$ is temporarily inactive, $q$ may finish its operation and leave $u$, and another particle $q'$ may occupy $u$. If new particles keep occupying $u$ before $p$ does, then $p$ will never complete its operation. Observe that, if a particle $q$ manages to occupy $u$, then it is able to finish its handedness communication operation. Indeed, $q$ will have to contract into $u$ and then go back to its original location. In turn, the original location of $q$ will necessarily be unoccupied, because the protocol prevents any particle from expanding into that vertex. Since no two particles perform a handedness communication operation together more than once, after a finite number of stages $p$ will not have to contend $u$ with any other particle, and will therefore be free to occupy it.

The second potential deadlock situation is similar: $p$ waits for some other particle $p'$ to contract or come back to its original location. If $p'$ keeps expanding to different locations to interact with other particles, $p$ will wait forever. Again, this situation is resolved by observing that, once $p'$ has expanded, it necessarily terminates its handedness communication operation. Also, $p'$ can only be involved in finitely many such operations.

The final potential deadlock situation is the following. A particle $p_1$ begins a handedness communication operation and locks an auxiliary particle $q_1$. However, the other particle that it intends to lock, $q_2$, is already locked by some other particle $p_2$. In turn, $p_2$ intends to lock another particle $q_3$ that is already locked, etc. If the $k$th particle in this chain, $p_k$, has locked $q_k$ but also wants to lock $q_1$, there is a deadlock. Observe that in a single tree of the spanning forest of $P$ there can be at most one handedness communication operation in progress. Since there are at most three such trees (because there are at most three candidate leaders), $k\leq 3$.
\begin{itemize}
\item If $k=1$, obviously there can be no deadlock.
\item If $k=2$, the sequence $(p_1,q_1,p_2,q_2)$ is a cycle in $G_D$ (with a little abuse of notation, we use particles' names to indicate the vertices they occupy). $p_1$ intends to communicate its handedness to its child, which is a neighbor of $p_1$, $q_1$, and $q_2$: therefore, it must be $p_2$. However, $p_2$ cannot be a child of $p_1$, because it lies in a different tree of the spanning forest.
\item If $k=3$, the sequence $(p_1,q_1,p_2,q_2,p_3,q_3)$ is a cycle in $G_D$. The only possibility is for these six particles to form a regular hexagon in $G_D$. Since the child of $p_1$ must be a neighbor of $p_1$, $q_1$, and $q_2$, it must occupy the center of the hexagon. Similarly, the same particle must be the child of $p_2$ and $p_3$, which is impossible, because a particle cannot have more than one parent.
\end{itemize}
In any case, there can be no deadlock.

Since no deadlocks can occur, eventually each non-Candidate particle is involved in a handedness communication operation with its parent, it learns if it has the same handedness as its parent, and it sends Done-Handedness-Communication messages to it. The initial agreement protocol among candidate leaders consists of a constant number of moves. Each handedness communication operation also consists of a constant number of moves, and exactly one such operation is performed for each non-Candidate particle of $P$. In total, at most $O(n)$ moves are performed in this phase. Moreover, whenever a particle moves, it then goes back to its original location before the phase is finished. It follows that $P$ forms $S_0$ again when the phase ends.

Similarly, since in a constant number of rounds a new particle either learns if it has the same handedness as its parent or forwards a Done-Handedness-Communication to its parent, the phase terminates in $O(n)$ rounds.
\end{proof}

\subsection{Leader Election Phase}\label{s:3.4}
When a candidate leader receives a Done-Handedness-Communication message from its last child, it knows that its entire tree has agreed on the same handedness. So, it transitions to the leader election phase. The goal of this phase is to elect a single leader among the candidates, if possible.

By Theorem~\ref{tp3}, at the end of the handedness agreement phase the system forms the initial shape $S_0$ again. In order to elect a leader, the candidates ``scan'' their respective trees of the spanning forest, searching for asymmetric features of $S_0$ that would allow them to decide which candidate should become the leader. This task is made possible by the fact that all particles agree on the same handedness. If no asymmetric features are found and no leader can be elected, then $S_0$ must be unbreakably $k$-symmetric, and the system will proceed to the next phase with $k$ leaders.

Technically, becoming a leader means setting an internal \emph{Leader} flag (which is initially not set) and clearing the Candidate flag.

\smallskip
\noindent\textbf{Neighborhood encoding.}
We preliminarily define a finite-length code $C(p)$ that a particle $p$ can use to describe its neighborhood to other particles. The code is a string of six characters from the alphabet $\{\mbox{L},\mbox{P},\mbox{C},\mbox{N}\}$. The $i$th character describes the content of the vertex $v$ such that $\widetilde\ell(p,v)=i$, where $\widetilde\ell$ is the corrected labeling of $p$ (refer to Section~\ref{s:3.3} for the definition of corrected labeling). The character is chosen as follows:
\begin{itemize}
\item L if $v$ is occupied by a candidate leader;
\item P if $v$ is occupied by the parent of $p$;
\item C if $v$ is occupied by a child of $p$;
\item N otherwise.
\end{itemize}
This information is readily available to $p$: indeed, at this point of the execution of the algorithm, $p$ is well aware of which of its neighboring vertices are occupied, where its parent is, where its children are, etc.

Note that, by Theorem~\ref{tp3}, using $\widetilde\ell$ in all particles' computations (as opposed to $\ell$) is equivalent to assuming that all particles have the same handedness (i.e., the handedness of the candidate leaders).

\smallskip
\noindent\textbf{Basic election technique.}
In the main leader election algorithm, the candidate leaders will repeatedly use the following ``tentative election procedure''.

Suppose that there are $k=2$ candidate leaders in $P$, namely $p_1$ and $p_2$. Let $p_1$ know the neighborhood code $C(q_1)$ of some particle $q_1$ in its tree. Similarly, $p_2$ knows the neighborhood code $C(q_2)$ of some particle $q_2$ in its tree. Then, $p_1$ sends $C(q_1)$ to $p_2$, and $p_2$ sends $C(q_2)$ to $p_1$. When they know both codes, they compare them. If $C(q_1)=C(q_2)$, the symmetry-breaking attempt fails, and the procedure ends. Otherwise, we can assume without loss of generality that $C(q_1)$, as a string, is lexicographically smaller than $C(q_2)$. So, $p_1$ becomes a Leader particle and the parent of $p_2$, while $p_2$ clears its Candidate and Eligible flags, and becomes a child of $p_1$. Of course, if $C(q_2)$ turns out to be lexicographically smaller, then $p_2$ becomes the Leader.

Suppose now that there are $k=3$ candidate leaders $p_1$, $p_2$, and $p_3$. Let each candidate $p_i$ know the neighborhood code $C(q_i)$ of some particle $q_i$ in its tree. As in the previous case, each candidate leader sends its code to the other two. When a candidate leader knows all three codes, it compares them. Without loss of generality, assume that $C(q_1)\leq C(q_2)\leq C(q_3)$ (lexicographically). There are three cases:
\begin{itemize}
\item If $C(q_1)<C(q_2)$, then $p_1$ becomes the unique Leader particle. $p_2$ and $p_3$ cease to be candidate leaders and become children of $p_1$.
\item If $C(q_1)=C(q_2)$ and $C(q_2)<C(q_3)$, then $p_3$ becomes the unique Leader particle. $p_1$ and $p_2$ cease to be candidate leaders and become children of $p_3$.
\item Otherwise, the three codes are equal, and the symmetry-breaking attempt fails.
\end{itemize}

\smallskip
\noindent\textbf{Main leader election algorithm.}
If there is only one candidate leader in $P$, it becomes a Leader particle, and the leader election phase ends there. So, let us assume that $P$ contains $k=2$ or $k=3$ pairwise adjacent candidate leaders.

Each candidate leader $p_i$ starts by sending its own neighborhood code $C(p_i)$ to the other candidate leaders, and the basic election procedure explained above is executed. If a Leader particle is elected, the phase ends.

If the election attempt fails, $p_i$ asks its first child $p'_i$ to fetch the neighborhood codes of the first particle in its subtree (i.e., $p'_i$ itself). When $p_i$ obtains this code, it uses it for another election attempt procedure. If the attempt fails, $p_i$ asks $p'_i$ for the code of another particle in its subtree, etc.

When $p'_i$ has exhausted its entire subtree, it sends a \emph{Subtree-Exhausted} message to $p_i$, which proceeds to querying its second child, and so on.

In turn, $p'_i$ and all other internal particles of the trees act similarly. When such a particle is instructed by its parent to fetch the neighborhood codes of the particles in its subtree, its starts with its own code, then queries its fist child, and the process continues recursively at all levels of the tree. When there are no more particles to query in the subtree, the particle sends a Subtree-Exhausted message to its parent.

If the candidate leader $p_i$ receives a Subtree-Exhausted message from its last child, and no leader has been elected, then $p_i$ becomes a Leader particle (as we will see in Theorem~\ref{tp4}, this means that $S_0$ is unbreakably $k$-symmetric).

\smallskip
\noindent\textbf{Canonical order of children.}
For this algorithm to work properly, we have to define a \emph{canonical order} in which a particle $p$ queries its children for their codes.

If $p$ is a candidate leader, its \emph{base neighbor} is defined as the unique candidate leader located in a vertex $v$ such that the port label $(\ell(p,v)+1)\mod 6$ does not correspond to a vertex occupied by another candidate leader. If $p$ is not a candidate leader, then its base neighbor is defined to be its parent.

The canonical order of the children of $p$ is the order in which they are found as $p$ scans its neighbors in clockwise order starting from its base neighbor. The ``clockwise order'' is defined according to the handedness of the candidate leaders, which $p$ is supposed to know, due to Theorem~\ref{tp3}.

\smallskip
\noindent\textbf{Synchronization.}
There is one last addition to make to the above protocol, which pertains to synchronization. Recall that, when a candidate leader $p_1$ obtains a code $C(q_1)$ from one of the particles in its tree, it sends it to the other candidate leaders. Then, $p_1$ waits until it has obtained codes from all other candidate leaders. Suppose that another candidate leader $p_2$ obtains the code $C(q_2)$ of a particle in its tree some stages after $p_1$. So, $p_2$ sends $C(q_2)$ to $p_1$ and receives $C(q_1)$ from it. Now $p_2$ has all the codes it needs, and it executes the election procedure, failing to elect a leader. Therefore, $p_2$ obtains a new code $C(q'_2)$ from another particle, and sends it to the other candidate leaders, including $p_1$. However, as $p_2$ was operating, the scheduler may have kept $p_1$ inactive: as a result, the message containing $C(q_2)$ was overwritten by the one containing $C(q'_2)$ before $p_1$ was able to read it. When $p_1$ is activated again, it compares $C(q_1)$ with $C(q'_2)$ (instead of $C(q_2)$), and it behaves incorrectly.

To avoid this desynchronization problem, we put a counter modulo 2 (i.e., a single bit) in the internal memory of each candidate leader. Whenever a candidate leader obtains a new code $C(q)$ from a particle in its tree, it increments the counter modulo 2 and it attaches its value to $C(q)$ before sending it to the other candidates.

Now, if a candidate leader $p_1$ receives a code with an unexpected counter bit from another candidate $p_2$, it implicitly knows that the previous election attempt has failed. In that case, $p_1$ obtains a new code from another particle in its tree, and proceeds with the protocol as usual.

On the other hand, if a leader is elected, there are no particular problems: as soon as a candidate leader $p$ realizes that an election procedure has succeeded, it transitions to the next phase, and communicates this information to the other candidate leaders, as explained at the beginning of Section~\ref{s:3}. While doing so, $p$ also adds information on who the Leader particle is, so that the other candidate leaders can change their internal variables consistently, even if they have failed to receive the last code from $p$.

\smallskip
\noindent\textbf{Correctness.}
\begin{theorem}\label{tp4}
Let $P$ be the system resulting from Theorem~\ref{tp3}, forming a shape $S_0$. If all particles of $P$ execute the leader election phase of the algorithm, then there is a stage $s$, reached after $O(n^2)$ rounds, where one of the two following conditions holds:
\begin{itemize}
\item There is a unique Leader particle in $P$, which is the root of a well-defined spanning tree of $P$.
\item There are $k=2$ or $k=3$ mutually adjacent Leader particles in $P$, and $S_0$ is unbreakably $k$-symmetric. Each Leader particle is the root of a well-defined tree: these $k$ trees collectively form a spanning forest of $P$ whose plane embedding has a $k$-fold rotational symmetry around the center of $S_0$.
\end{itemize}
At stage $s$, all non-Leader particles are non-Eligible. No particle moves in this phase.
\end{theorem}
\begin{proof}
Assume there are $k=2$ or $k=3$ candidate leaders at the beginning of this phase, because otherwise the theorem is trivial. Let $c$ be the center of the subsystem formed by the Candidate particles, and let $\rho$ be the $k$-fold rotation around $c$.

It is easy to prove by induction that the candidate leaders perform several tentative election procedures, each time with the neighborhood code of a new particle in their respective tree, until a Leader is elected or no more particles are left in the tree of some candidate leader. Moreover, the fact that all particles agree on the clockwise direction, the way the canonical order of children is defined, and the information contained in the neighborhood codes imply that the candidate leaders will always compare the codes of particles that are symmetric under $\rho$, until asymmetric particles are found.

So, a Leader particle will definitely be elected if there are two particles in the trees of two different candidate leaders whose neighborhoods are not symmetric under $\rho$. This necessarily happens if the trees of the spanning forest are not symmetric under $\rho$ (and it happens before a candidate leader runs out of particles in its tree). If the trees are symmetric, then in particular they have the same size, and $S_0$ is unbreakably $k$-symmetric. In this case, particles that are symmetric under $\rho$ may still produce different codes (because, in general, their codes are ``rotationally equivalent'', but not necessarily identical), which results in the election of a Leader. Otherwise, all election attempts will fail, the candidate leaders will run out of particles at the same time, and they will all become Leaders.

Since the phase terminates after $O(n)$ election attempts, each of which lasts $O(n)$ rounds, the whole phase takes $O(n^2)$ rounds.
\end{proof}

\subsection{Straightening Phase}\label{s:3.5}
At the beginning of this phase, there are $k=1$, $k=2$, or $k=3$ Leaders, each of which is the root of a tree of particles. These $k$ trees are rotated copies of each other, and the Leaders are pairwise adjacent.

The goal of each Leader is to coordinate the ``straightening'' of its tree. That is, in the final stage of this phase, the system must form $k$ straight line segments, each of which has a Leader located at an endpoint. Moreover, if $k>1$, each Leader must also lie on the extension of another of the $k$ segments.

\smallskip
\noindent\textbf{Choosing the directrices.}
Each Leader $p_i$ will choose a ray in the plane (i.e., a half-line) as its \emph{directrix} $\gamma_i$. By the end of the straightening phase, all particles will be located on these $k$ directrices.

If $k=1$, the unique Leader $p_1$ arbitrarily chooses a neighboring vertex $v$, and picks the ray from $p$ through $v$ as its directrix $\gamma_1$.

If $k=2$, there are two adjacent Leaders $p_1$ and $p_2$. $p_1$ chooses its neighbor $v$ that is opposite to $p_2$, and the ray from $p_1$ through $v$ is its directrix $\gamma_1$. On the other hand, $p_2$ chooses the symmetric ray as its directrix $\gamma_2$.

If $k=3$, there are three pairwise adjacent Leaders $p_1$, $p_2$, and $p_3$. Each Leader $p_i$ picks its ``left'' neighboring Leader $p_j$ (according to its handedness), and lets $v_i$ be its neighboring vertex that is opposite to $p_j$. Then, $p_i$ defines its directrix $\gamma_i$ to be the ray from $p_i$ through $v_i$. Since the Leaders have the same handedness (see Theorem~\ref{tp3}), their three directrices are pairwise disjoint and form angles of $120^\circ$ with each other.

\smallskip
\noindent\textbf{Basic pulling procedure.}
For this sub-protocol, we assume to have a linearly ordered \emph{chain} of particles $Q\subseteq P$, the first of which is called the \emph{Pioneer} particle. Each particle of $Q$ except the Pioneer has a unique \emph{Predecessor} in $Q$, located in an adjacent vertex of $G_D$. Similarly, each particle except the last one has a unique \emph{Follower} in $Q$, located in an adjacent particle. All particles of $Q$ are initially contracted.

Say the Pioneer particle $q$ intends to move into a neighboring unoccupied vertex $v$, which is called its \emph{destination}. The \emph{pulling} procedure will make $q$ move into $v$, and will subsequently make each Follower move into the vertex previously occupied by it Predecessor. At the end of the procedure, the particles of $Q$ will still form a chain with the same Follower-Predecessor relationships, and all particles will be contracted.

To begin with, $q$ sends a \emph{Follow-Me} message to its Follower $q'$, and then it expands toward $v$ and contracts again in $v$. $q'$ will read the message from $q$ and will send a similar Follow-Me message to its Follower $q''$. Then, as soon as $q'$ sees that the original location of $q$ is empty, it expands toward it and contracts again.

The procedure continues in this fashion until the last particle of $Q$ has moved and contracted into its Predecessor's original location. At this point, the last particle sends a \emph{Movement-Done} message to its Predecessor, which reads it and forwards it to its Predecessor, and so on. When the Pioneer receives a Movement-Done message, the procedure ends.

\smallskip
\noindent\textbf{Main straightening algorithm.}
The idea of this phase is that each Leader $p_i$ will identify a directrix $\gamma_i$ (as explained above), and a Pioneer $q_i$ will walk along $\gamma_i$, pulling particles onto it from the tree $T_i$ of $p_i$ (executing the pulling procedure described above). While the Pioneer is doing that, the Leader remains in place, except perhaps for a few stages, when it is part of a chain of particles that is being pulled by the Pioneer. Eventually, all the particles of $T_i$ will form a line segment on the directrix, and the Leader will be at an endpoint of such a segment, opposite to the Pioneer.

If $q_i$ encounters another particle $r$ on $\gamma_i$, belonging to some tree $T_j$, it ``transfers'' its role to $r$, and ``claims'' the subtree $T'_j$ of $T_j$ hanging from $r$, detaching $r$ from its parent. The next time the new Pioneer $r$ has to pull a chain of particles, it will pull it from $T'_j$. For this reason, $r$ is called an \emph{entry point} of the directrix. This algorithm is summarized in Figure~\ref{f:straight}.

\begin{figure}[ht!]
\begin{center}
  \subfloat[Each Pioneer is obstructed by a particle on its directrix.]{
  \includegraphics[width=0.475\textwidth]{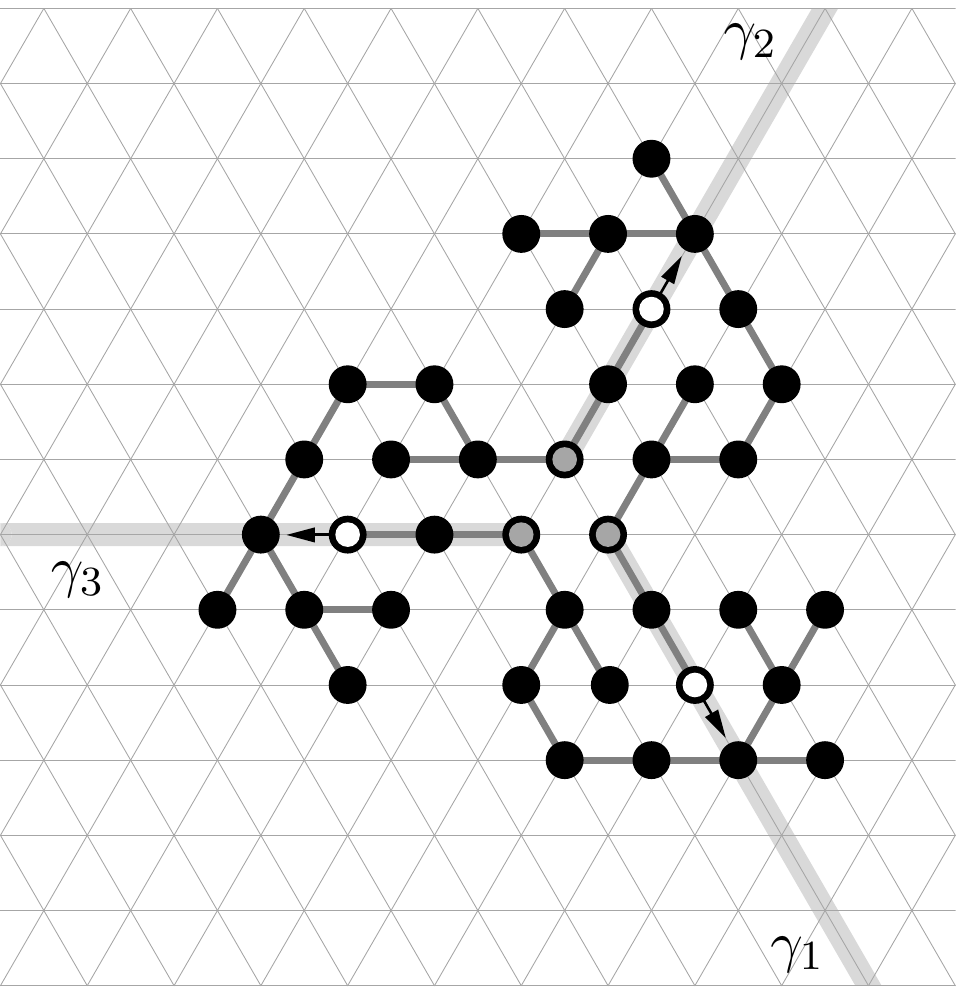}
  \label{f:straight1}
  }
  \hfill
  \subfloat[The obstructing particles detach from their Parents and become the new Pioneers.]{
  \includegraphics[width=0.475\textwidth]{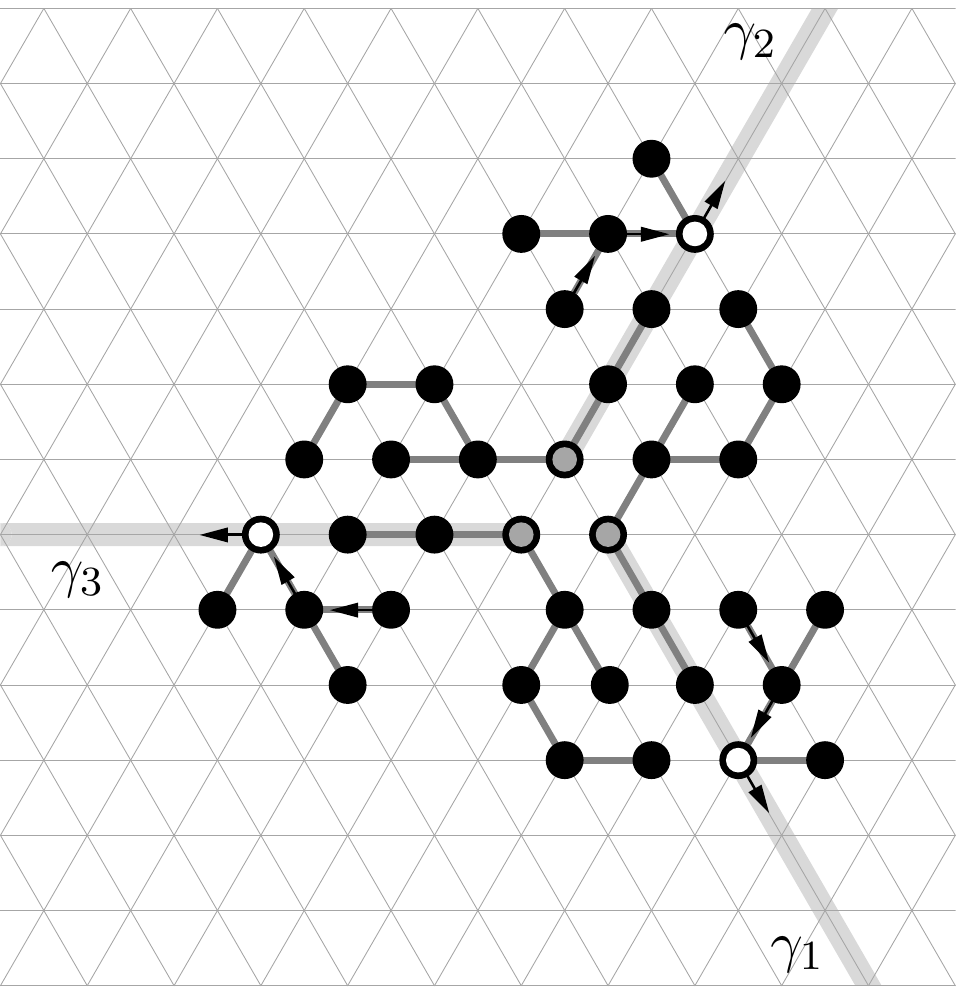}
  \label{f:straight2}
  }
  \\
  \subfloat[Each new Pioneer pulls a chain of particles from the closest entry point.]{
  \includegraphics[width=0.475\textwidth]{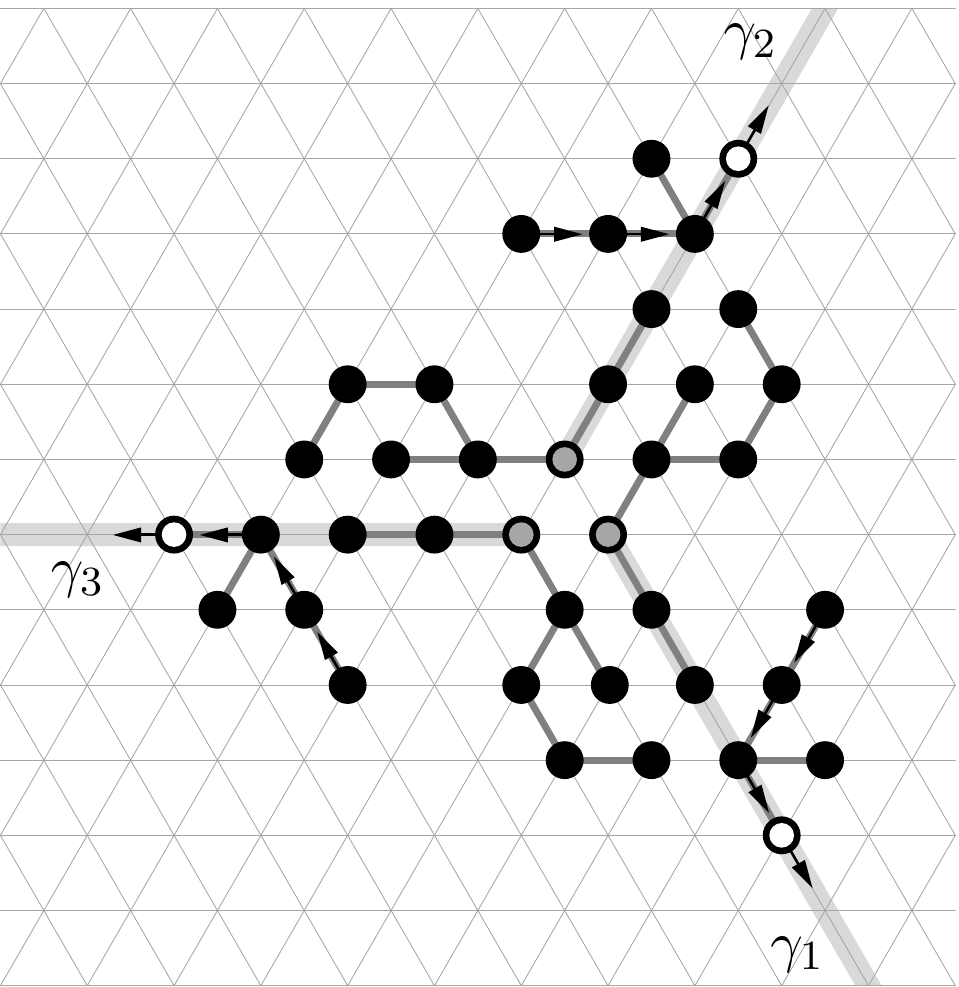}
  \label{f:straight3}
  }
\end{center}
\caption{Three stages of the straightening phase. The particles in gray are the Leaders; the ones in white are the Pioneers. The edges of the spanning forest are drawn in dark gray, and the arrows indicate where the particles are directed in the pulling procedure.}
\label{f:straight}
\end{figure}

Every time a Pioneer advances along its directrix, it notifies its Leader, who will synchronize with the other Leaders. This is to ensure that the straightening of every tree proceeds at the same pace.

Technically, a Pioneer is identified by an internal \emph{Pioneer} flag, and an entry point is identified by an \emph{Entry-Point} flag. Both flags are initially not set. As the phase starts, each Leader sets its own Pioneer and Entry-Point flags. Then, the following operations are repeated until the end of the phase.
\begin{itemize}
\item The algorithm works in \emph{steps}: when a Leader decides to start a new step, it sends a \emph{Pull} message to the next particle on its directrix. This message is forwarded by the particles along the directrix, until it reaches the Pioneer (of course, if the Leader and the Pioneer are the same particle, no message is actually sent).
\item When the Pioneer $q$ receives the Pull message, it looks at the next vertex $v$ along the directrix (i.e., in the direction opposite to the Leader). Suppose first that $v$ is occupied by a (contracted) particle $r$.
\begin{itemize}
\item $q$ sends a \emph{You-are-a-Pioneer} message to $r$ and clears its own Pioneer flag.
\item $r$ reads the message and becomes a Pioneer. If $r$ has children, it also sets its own Entry-Point flag, becoming an entry point.
\item $r$ informs its parent $r'$ that it is no longer its child, and erases its own Parent variable.
\item $r'$ modifies its internal variables accordingly and sends a message back to $r$.
\item When the new Pioneer $r$ receives this message, it proceeds with the algorithm.
\end{itemize}
\item Suppose now that $v$ (i.e., the next vertex along the directrix) is unoccupied.
\begin{itemize}
\item The Pioneer $q$ starts executing the pulling procedure with destination $v$.
\item When $q$ reaches $v$, it makes sure that its Entry-Point flag is not set.
\item When the Follow-Me message that is forwarded along the directrix reaches the first entry point $e$ (possibly, $e$ is the Pioneer itself), $e$ forwards the message to its first child according to the canonical order defined in Section~\ref{s:3.4}. Similarly, whenever a particle in the subtree $T$ hanging from $e$ receives a Follow-Me message from its parent, it forwards it to its first child according to the canonical order. 
\item When $e$ moves, it also clears its Entry-Point flag and sends a \emph{You-are-an-Entry-Point} message to its Follower $e'$ (i.e., its first child).
\item If $e'$ has children, it sets its own Entry-Point flag upon receiving this message. Otherwise, $e'$ does not become an entry point, and claims the next particle on the directrix as its Follower (if such a particle exists).
\item If $e$ is a Leader, it clears its Leader flag, while $e'$ sets its own and becomes the new Leader.
\item When a particle $t$ of $T$ sends a Follow-Me message to its first child $t'$, it also attaches its neighborhood code $C(t)$ to the message, as defined in Section~\ref{s:3.4}. $t'$ memorizes the code.
\item When $t'$ moves to take the place of $t$, it updates its internal variables according to $C(t)$. Of course, if $t'$ was a leaf of $T$, it does not include its previous location in the list of its children.
\item When a leaf of $T$ moves to take the place of its Predecessor, it sends a Movement-Done message to its parent, which is forwarded to the Pioneer. When the Pioneer receives the message, it proceeds with the algorithm.
\end{itemize}
\item The Pioneer sends a \emph{More-Entry-Points?} message along its directrix, which is forwarded by the particles lying on it, until it reaches the Leader.
\item As the particles (including the Pioneer) forward the More-Entry-Points? message, they add information to it, i.e., they set a flag in the message if they are entry points of the directrix.
\item When the Leader reads the More-Entry-Points? message, it knows if the phase is over (i.e., there are no more entry points on its directrix), or if it has to start another step (i.e., the Pioneer has to pull more particles).
\item If the phase is not over, the Leader $p$ synchronizes with the other Leaders (of course, if $k=1$, this step is skipped). If any of the other Leaders is not found in its usual position (because it is still executing a pulling procedure and is being replaced by a new Leader), $p$ waits for the new Leader to appear. The actual synchronization is done by exchanging \emph{Next-Straightening-Step} messages, together with the value of a counter modulo 2, as described in the ``Synchronization'' paragraph of Section~\ref{s:3.4}. When $p$ receives such messages from all other Leaders, its starts the next step. 
\end{itemize}

\smallskip
\noindent\textbf{Correctness.}
\begin{theorem}\label{tp5}
Let $P$ be the system resulting from Theorem~\ref{tp4}, with $k$ Leader particles. If all particles of $P$ execute the straightening phase of the algorithm, then there is a stage, reached after $O(n^2)$ rounds, where all particles are contracted, the $k$ Leaders are pairwise adjacent, and the system forms $k$ equally long straight line segments, each of which has a Leader located at an endpoint. Moreover, if $k>1$, each such segment has a second Leader lying on its extension. In this phase, at most $O(n^2)$ moves are performed in total.
\end{theorem}
\begin{proof}
Suppose that $k>1$. By assumption, as the phase starts, the spanning forest of $P$ is symmetric under a $k$-fold rotation of the plane. Also, each tree of the forest is attached to a directrix by an entry point (initially, only the Leaders are entry points). We can easily prove by induction that these properties are preserved after each step of the algorithm.

This is because the $k$ Leaders wait for each other at the end of every step, until they are all ready to start the next step. Moreover, if a Pioneer encounters a vertex occupied by a particle $r$ on its directrix, then so do all other Pioneers, and vice versa. Additionally, the subtree hanging from $r$ is symmetric to the ones that are hanging from the particles that are encountered by the other $k-1$ Pioneers during the same step.

Note that, when $r$ becomes an entry point, the tree to which it belongs splits in two, because $r$ is detached from its parent, and the whole subtree hanging from $r$ is attached to the directrix, as $r$ becomes a new entry point. However, this keeps the structure connected.

On the other hand, when a pulling procedure is executed, all the particles in the chain that belong to a tree choose their Follower according to the canonical order of their children. Hence, as $k$ pulling procedures are executed by the $k$ Pioneers and their chains during a step, symmetric particles on different chains move in symmetric ways, and the overall symmetry of the system is preserved. Again, this keeps the structure connected.

Because of this symmetry, no conflicts between different Pioneers can ever arise. For instance, it is impossible for a leaf $f$ of a tree to be pulled along the chain led by a Pioneer while another Pioneer is sending a You-are-a-Pioneer message to $f$. Also, the $k$ pulling procedures that are executed in the same step involve disjoint chains: indeed, the $k$ directrices are disjoint, and the subtrees hanging from different entry points are disjoint.

Also observe that, even as the chains move, no messages are ever lost. This is because, at any time, at most $k$ independent pulling procedures are being executed. In each pulling procedure, every time a message is sent, the addressee is a still and contracted particle that necessarily receives and reads the message as soon as it is activated. Additionally, if a Leader is not in place because it is being substituted by its Follower, the other Leaders do not send synchronization messages its way, but wait until the new Leader is in position.

Hence, each step correctly terminates and results in the advancement of every Pioneer and the addition of a new particle to every directrix. If the Leaders order the beginning of a new step, it is because the More-Entry-Points? messages have revealed the presence of more entry points on the directrices. The straightening phase only ends when no more entry points are found: since the structure is connected, this means that all particles are indeed aligned on the directrices, forming $k$ line segments. By the symmetry of the system, these line segments must have the same length.

To prove that at most $O(n^2)$ moves are made in total, observe that each pulling procedure causes a new particle to join the portion of a directrix located between a Leader and a Pioneer. The particles located in this portion never leave the directrix, but only move along it. So, at most $n$ pulling procedures are performed. Also, each pulling procedure involves at most $n$ particles, and causes each of them to perform a single expansion and a single contraction. Since no other moves are made by the system, the $O(n^2)$ bound follows.

Similarly, since a pulling procedure is completed after $O(n)$ rounds, the whole straightening phase takes $O(n^2)$ rounds.
\end{proof}

\subsection{Role Assignment Phase}\label{s:3.6}
At the end of the straightening phase, the system forms $k$ equally long line segments, arranged as described in Theorem~\ref{tp5}, each of which contains a Leader particle. If $k>1$, it means that the shape $S_0$ that the particles originally formed was unbreakably $k$-symmetric, as Theorems~\ref{tp1}--\ref{tp5} summarize. Due to Theorem~\ref{t:neg}, if this is the case, we have to assume that the ``final shape'' $S_F$ that the system has to form is also unbreakably $k$-symmetric.

Recall that a representation of $S_F$ is given to all the particles as input, and resides in their internal memory since the first stage of the execution. For the purpose of the universal shape formation algorithm, we assume the size of $S_F$ to be a constant with respect to the number of particles in the system, $n$ (cf.~Section~\ref{s:2}). Also, we may assume $S_F$ to be a minimal shape: if it is not, the particles replace its representation with that of a minimal shape equivalent to $S_F$, which has a smaller size and is readily computable. Finally, since the handedness agreement phase has been completed, all the particles can be assumed to have the same handedness (see Theorem~\ref{tp3}). Without loss of generality, we assume that their notion of clockwise direction coincides with the ``correct'' one, i.e., the one defined by the cross product of vectors in $\mathbb R^2$.

The goal of the role assignment phase is twofold:
\begin{itemize}
\item The particles determine the scale of the shape $S'_F$, equivalent to $S_F$, that they are going to form. Indeed, if $n$ is large enough, there is a scaled-up copy of $S_F$ that can be formed by exactly $n$ particles, keeping in mind that, in the final configuration, particles can be contracted or expanded. 
\item Each particle is assigned a constant-size identifier, describing which element of $S_F$ (i.e., a vertex, the interior of an edge, or the interior of a triangle) the particle is going to form in the shape composition phase. Recall that we are assuming $S_F$ to be composed of a constant (i.e., independent of $n$) number of triangles and edges, in accordance to the definition of universal shape formation (see Section~\ref{s:2}). The size of the identifier is proportional to the size of $S_F$, and can therefore be stored in a single particle's internal memory.
\end{itemize}

\smallskip
\noindent\textbf{Subdividing the final shape into elements.}
Recall that a shape is the union of finitely many edges and faces of $G_D$. Of course, all edges of $G_D$ have length $1$, and all faces of $G_D$ are equilateral triangles of side length $1$. Let the final shape $S_F$ be of the form $S_F=e_1\cup\dots\cup e_j\cup t_1\cup\dots\cup t_{j'}$, where the $e_i$'s are edges of $G_D$ and the $t_i$'s are (triangular) faces of $G_D$.

Let $S'_F$ be a shape equivalent to $S_F$. By Lemma~\ref{l:shapescale}, the scale of $S'_F$ is a positive integer $\lambda$ (recall that $S_F$ is minimal). That is, there is a similarity transformation $\sigma\colon\mathbb R^2\to\mathbb R^2$ such that $\sigma(e_i)$ is a segment of length $\lambda$ (i.e., it is the union of $\lambda$ consecutive segments of $G_D$) contained in $S'_F$ and $\sigma(t_i)$ is an equilateral triangle of side length $\lambda$ contained in $S'_F$.

Let $B$ be the set of vertices of $G_D$ that are contained in $S'_F$. We partition $B$ into three families of \emph{elements} as follows:
\begin{itemize}
\item If $v$ is a vertex of $G_D$ contained in $S_F$, then $\sigma(v)$ constitutes a \emph{super-vertex} of $S'_F$.
\item For every $e_i$, the vertices of $G_D$ that are contained in $\sigma(e_i)$ and are not super-vertices of $S'_F$ constitute a \emph{super-edge} of $S'_F$. Similarly, for every side $s$ of every triangle $t_i$, the vertices of $G_D$ that are contained in $\sigma(s)$ and are not in a super-vertex of $S'_F$ constitute a super-edge of $S'_F$.
\item For every $t_i$, the vertices of $G_D$ that are contained in $\sigma(t_i)$ and are not super-vertices of $S'_F$ or contained in super-edges of $S'_F$ constitute a \emph{super-triangle} of $S'_F$.
\end{itemize}
Observe that every super-vertex of $S'_F$ is a vertex of $G_D$, every super-edge of $S'_F$ is a set of $\lambda-1$ consecutive vertices of $G_D$, and every super-triangle of $S'_F$ is a set of $(\lambda-1)(\lambda-2)/2$ vertices of $G_D$ whose convex hull is an equilateral triangle of side length $\max\{0,\lambda-3\}$.

\begin{figure}[ht!]
\begin{center}
  \subfloat[An unbreakably $2$-symmetric shape with scale $5$ with a minimal equivalent shape consisting of two adjacent faces and two dangling edges]{
  \includegraphics[width=0.65\textwidth]{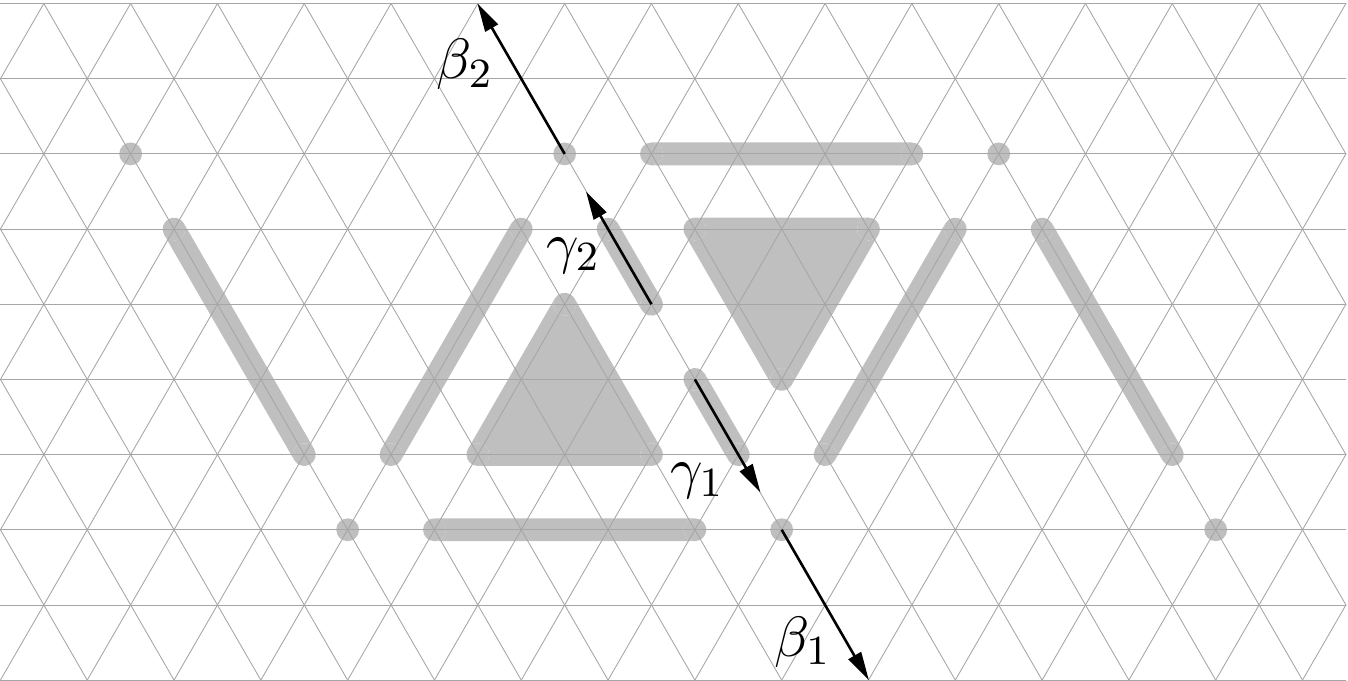}
  \label{f:elements1}
  }
  \\
  \subfloat[An unbreakably $3$-symmetric shape with scale $13$ with a minimal equivalent shape consisting of a single face]{
  \includegraphics[width=0.65\textwidth]{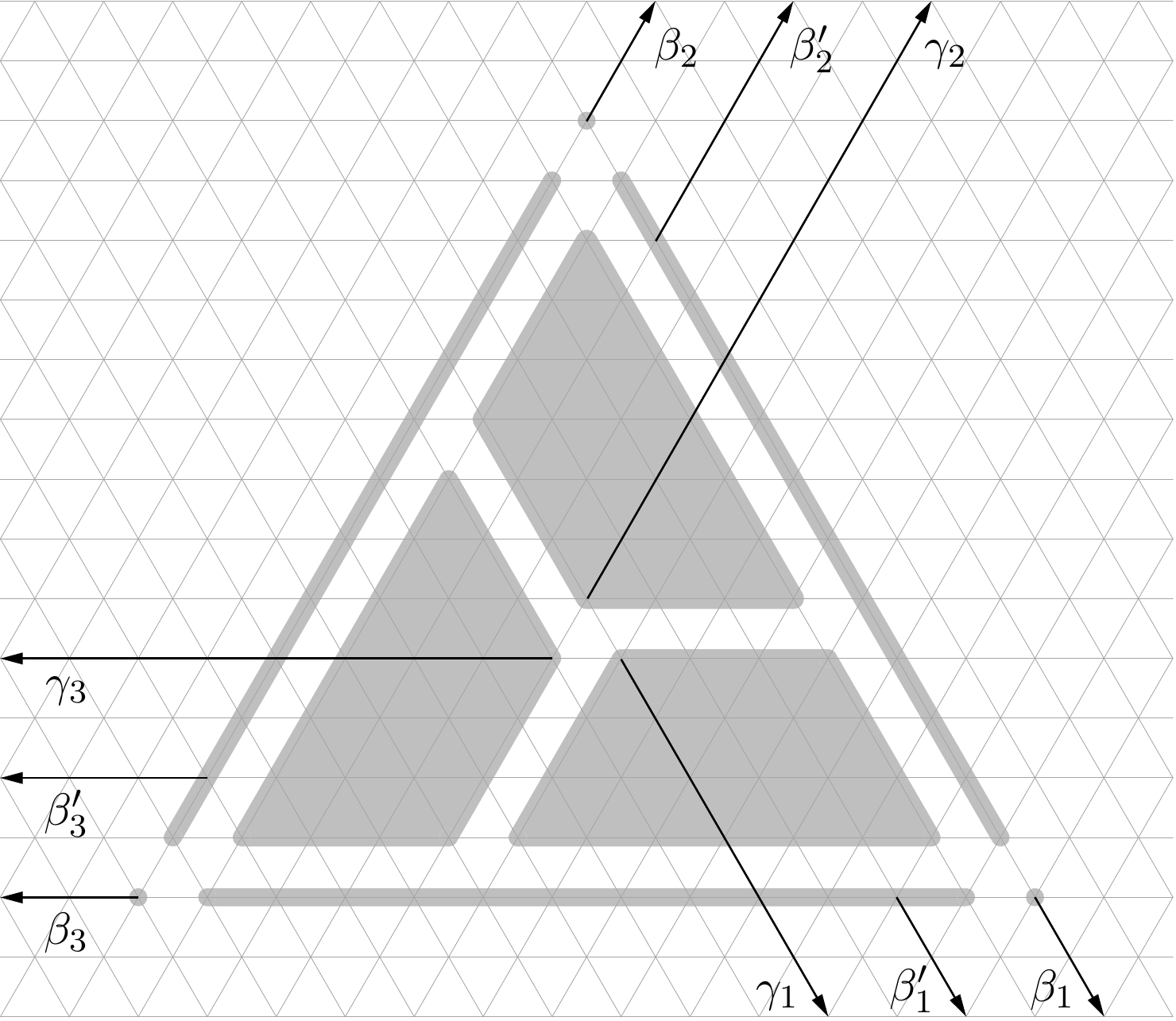}
  \label{f:elements2}
  }
\end{center}
\caption{Subdivision into elements (gray blobs) of unbreakably $k$-symmetric shapes. The directrices, the backbone, and the co-backbone are also represented.}
\label{f:elements}
\end{figure}

There is a small exception to our previous definition of element. Suppose that the system executing the role assignment phase has $k>1$ Leaders, and therefore $S_F$ is unbreakably $k$-symmetric. Suppose that $S'_F$ is unbreakably $k$-symmetric, as well: so, by Lemma~\ref{l:shapemin}, $\lambda$ is not a multiple of $k$. Finally, suppose that $S'_F$ contains its own center. We have two cases:
\begin{itemize}
\item If $k=2$, then the center of $S'_F$ is located in the midpoint of a super-edge $e$ consisting of an even number of vertices of $G_D$ (see Figure~\ref{f:elements1}). So, $e$ is divided by its midpoint into two \emph{partial super-edges} $e'$ and $e''$. In this case, $e$ is not an element of $S'_F$, but $e'$ and $e''$ are.
\item If $k=3$, then the center $c$ of $S'_F$ is located in the center of a super-triangle $t$ consisting of a number of vertices of $G_D$ that is a multiple of $3$ (see Figure~\ref{f:elements2}). Let $v_1$, $v_2$, $v_3$ be the vertices of $t$, taken in counterclockwise order. Let $\zeta$ be the ray emanating from $c$ in the direction of the vector $\overrightarrow{v_1v_2}$, and let $\zeta'$ and $\zeta''$ be the two rays emanating from $c$ and forming angles of $120^\circ$ with $\zeta$. These three rays partition $t$ into three symmetric \emph{partial super-triangles} $t'$, $t''$, and $t'''$ (note that no vertex of $G_D$ lies on any of these rays). In this case, $t$ is not an element of $S'_F$, but $t'$, $t''$, and $t'''$ are.
\end{itemize}
Observe that the set $\zeta\cup\zeta'\cup\zeta''$ has a 3-fold rotational symmetry, and so does the partition of $t$ into the elements $t'$, $t''$, and $t'''$. Given $t$, different particles may disagree on which vertex is $v_1$ and which vertex is $v_2$, and thus they may disagree on the orientation of $\zeta$. However, since all particles have the same handedness, they agree on the clockwise direction: so they agree on $\zeta\cup\zeta'\cup\zeta''$, and therefore also on the partition of $t$ into elements.

We denote by $m_v$ the number of super-vertices of $S'_F$, by $m_e$ the number of its super-edges, and by $m_t$ the number of its super-triangles. Of course, these numbers are independent of the scale of $S'_F$, and only depend on $S_F$.

By definition, forming $S'_F$ means occupying all vertices of $B$ with particles. This is equivalent to forming all super-vertices, all (partial) super-edges, and all (partial) super-triangles of $S'_F$, i.e., all the elements of $S'_F$.

\smallskip
\noindent\textbf{Combinatorial adjacency between elements.}
For the next part of the algorithm, we have to define a symmetric \emph{combinatorial adjacency} relation between elements of $S'_F$. This is slightly different from the relation induced by the neighborhood of the vertices of $G_D$ that constitute the elements.

The combinatorial adjacency rules are as follows:
\begin{itemize}
\item A super-vertex located in a vertex $v$ of $G_D$ and a (partial) super-edge $e$ are combinatorially adjacent if $e$ has an endpoint that neighbors $v$.
\item A (partial) super-edge $e$ and a (partial) super-triangle $t$ are combinatorially adjacent if every vertex of $G_D$ that is in $e$ has a neighbor in $t$.
\end{itemize}

Note that this relation induces a bipartite graph on the elements of $S'_F$: combinatorial adjacency only holds between a (partial) super-edge and a super-vertex or a (partial) super-triangle, and never between elements of the same kind or between super-vertices or (partial) super-triangles.

\smallskip
\noindent\textbf{Subdividing the elements among Leaders.}
Suppose that there are $k>1$ Leaders, and the shape $S'_F$ (similar to $S_F$) is unbreakably $k$-symmetric. We are going to show how each Leader selects the elements of $S'_F$ that the particles on its directrix will form in the shape composition phase. The result of this selection is exemplified in Figure~\ref{f:selection}.

Let $\sigma$ be a similarity transformation that maps $S_F$ to $S'_F$, and let $c$ be the center of $S_F$. We will assume the scale of $S'_F$ to be $\lambda\geq 4$.

We will first define $k$ rays, called the \emph{backbone} of $S'_F$ (see Figure~\ref{f:elements}). The backbone is an important structure that will be used extensively in the shape composition phase of the algorithm. Additionally, if $k=3$, we will also define a \emph{co-backbone} of $S'_F$, which is another set of $k$ rays that will only be used in the present paragraph. The definitions are as follows.
\begin{itemize}
\item If $k=2$, then $c$ is located in the midpoint of an edge $e$ of $G_D$. Let $v_1$ and $v_2$ be the endpoints of the segment $\sigma(e)$. One ray $\beta_1$ of the backbone is defined as the ray emanating from $v_1$ in the direction opposite to $v_2$. The other ray $\beta_2$ of the backbone emanates from $v_2$ in the direction opposite to $v_1$.
\item If $k=3$, then $c$ is located in the center of a triangular face $t$ of $G_D$. Let $v_1$, $v_2$, $v_3$ be the vertices of the triangle $\sigma(t)$, taken in counterclockwise order. One ray $\beta_1$ of the backbone is defined as the ray emanating from $v_1$ in the direction opposite to $v_2$. Similarly, the second ray $\beta_2$ emanates from $v_2$ in the direction opposite to $v_3$, and the third ray $\beta_3$ emanates from $v_3$ in the direction opposite to $v_1$.

The first ray $\beta'_1$ of the co-backbone is obtained by translating $\beta_1$ by $2$ in the direction parallel to the vector $\overrightarrow{v_1v_3}$. Similarly, the second ray $\beta'_2$ is obtained by translating $\beta_2$ by $2$ in the direction parallel to the vector $\overrightarrow{v_2v_1}$, and the third ray $\beta'_3$ is obtained by translating $\beta_3$ by $2$ in the direction parallel to the vector $\overrightarrow{v_3v_2}$.
\end{itemize}

Suppose that there is a bijection between Leader particles and rays of the backbone upon which all particles agree. Without loss of generality, let us say that the Leader $p_i$ ``claims'' the ray $\beta_i$ of the backbone of $S'_F$, for $1\leq i\leq k$. If $k=3$, the Leader $p_i$ also claims the ray $\beta'_i$ of the co-backbone.

Now, each Leader $p_i$ selects the elements of $S'_F$ that are fully contained in its own ray $\beta_i$ of the backbone. If $k=3$, then $p_i$ also selects the elements of $S'_F$ that have a non-empty intersection with its own ray $\beta'_i$ of the co-backbone (recall that $\lambda\geq 4$, hence the super-edges consist of at least two points, and the super-triangles consist of at least one point). Furthermore, $p_i$ selects the unique partial super-edge or partial super-triangle of $S'_F$ (depending on whether $k=2$ or $k=3$) that is closest to $\beta_i$.

Then, each Leader repeatedly selects an element of $S'_F$ that is combinatorially adjacent to an element that it has already selected and that has not been selected by any Leader, yet. While doing so, it makes sure that, if it has selected a super-vertex located on the backbone, then it also selects a (partial) super-edge that is combinatorially adjacent to it (in other words, there must be no ``isolated'' super-vertices in its selection). It is easy to see that selecting elements in this fashion is always possible.

\begin{figure}[ht!]
\begin{center}
\includegraphics[scale=0.7]{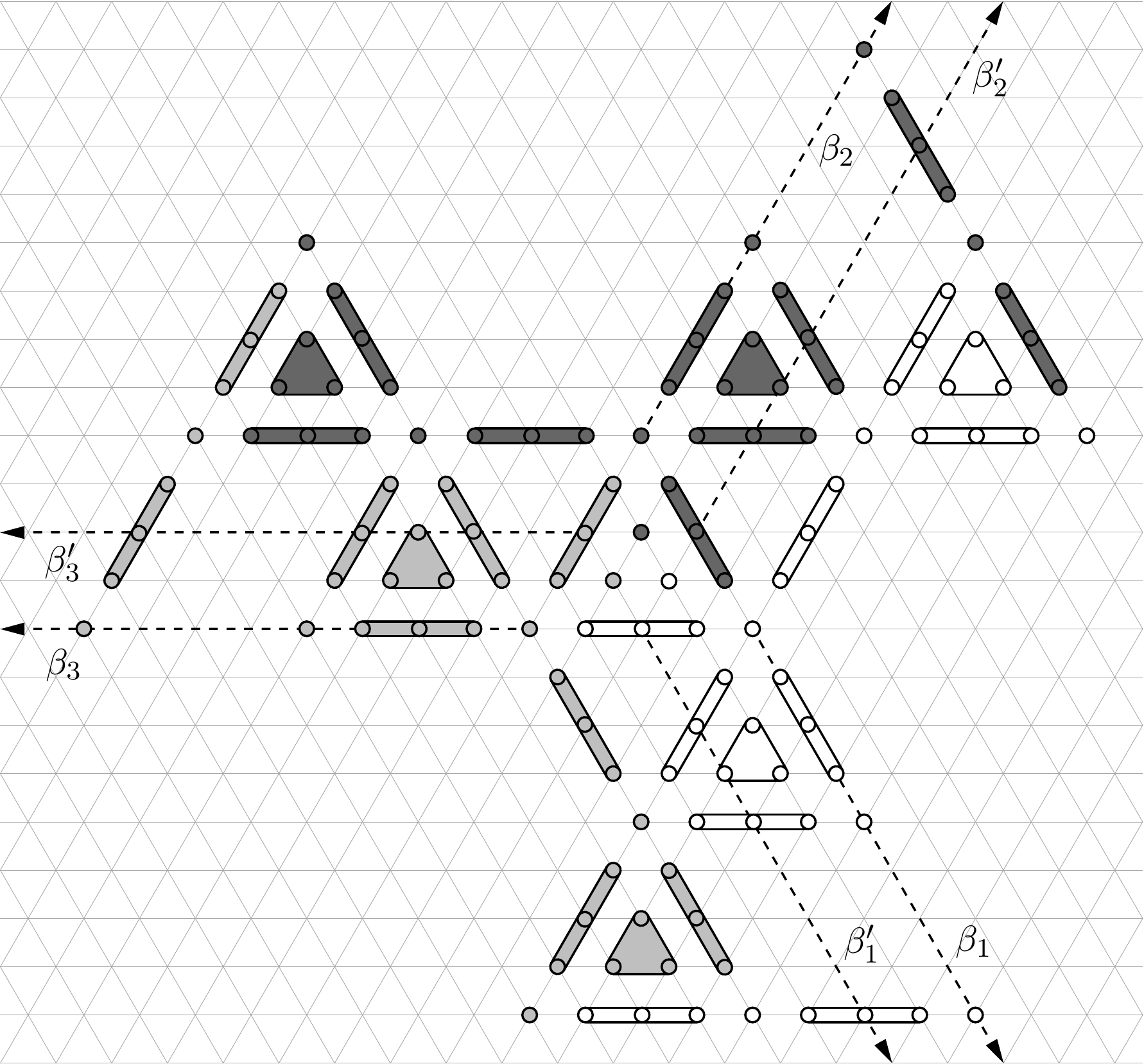}
\end{center}
\caption{The elements of an unbreakably $3$-symmetric shape with a possible subdivision among leaders. Blobs of the same color represent elements selected by the same leader.}
\label{f:selection}
\end{figure}

The actual selection algorithm is not important, as long as it is deterministic and only depends on the combinatorial adjacency relation between elements of $S'_F$. Note that, since the algorithm is deterministic and $S'_F$ is rotationally symmetric, the selections that the Leaders make are symmetric, too. In particular, the Leaders agree on each other's selections, and eventually $S'_F$ is divided into $k$ symmetric regions, each belonging to a different Leader.

Recall that $S_F$ has constant size, and hence $S'_F$ has a constant number of elements inducing a combinatorial adjacency relation of constant size. If the selection algorithm only depends on the combinatorial structure of the elements of $S'_F$ and on their local spatial layout (e.g., how different super-edges adjacent to the same super-vertex are laid out around it), then the algorithm can be executed internally by any particle in a single stage, even with its limited memory capabilities.

\smallskip
\noindent\textbf{Turing machine analogy.}
Let us focus on a single directrix. As the phase proceeds, the Leader of this directrix will ``walk'' along it, ``updating'' the states of the particles it encounters. Obviously, the Leader cannot physically move through another particle, but it will rather send it an \emph{I-am-Moving-to-your-Location} message. Then, the Leader will clear its own Leader flag, and the particle that receives the message will set its own. In other words, particles stay still, and the leadership is transferred from a particle to a neighboring one along the directrix.

Additional information can be attached to the I-am-Moving-to-your-Location message, containing a constant-size ``virtual internal state'' of the Leader. So, a particle that is hosting the Leader has its own internal state (as usual), but can also access and update the virtual internal state of the Leader (through an exchange of messages). In the following, to help intuition, we will pretend that the Leader is not a virtual particle with a virtual state, but a special particle with its own state that can walk through other particles.

As the Leader walks along the directrix, it updates its own internal state, as well as the states of the particles it encounters, much like the head of a Turing machine does as it scans the cells of a tape. So, using the states of the particles located on the directrix, the Leader can compute any function that is computable by a deterministic Turing machine on an blank tape of $n/k$ cells. This already gives the Leader an arsenal of subroutines and techniques with which it can operate on the states of the particles.

Say that the Leader wants to perform the same operation on a row of $j$ particles, where $j$ is too large to be stored in the Leader's internal memory. Suppose that the Leader has already constructed a representation of $j$, such as a binary code that fits in the states of the first $O(\log j)$ particles of the directrix. Then, the Leader can come back and decrement this number every time it operates on a particle. When the counter reaches $0$, the Leader knows it has to stop.

With this technique, the Leader can also ``shift'' the states of an entire row of $j$ particles by $j'$ positions to the left or to the right along the directrix, provided that the numbers $j$ and $j'$ are represented in binary in the states of a few particles. It can ``swap'' the states of two rows of $j$ particles, etc.

If the Leader has represented two numbers $a$ and $b$ in binary, it can easily compute the binary representation of their sum or their product with standard techniques, provided that the total number of particles on the directrix is at least $O(\log(a+b))$. Furthermore, given the representation of a number $x$, the Leader can compute any polynomial function of $x$ with constant coefficients, provided that there are $O(\log(x))$ particles on the directrix.

\smallskip
\noindent\textbf{Linearization of the elements.}
In the role assignment phase, the particles will use a new internal variable, called \emph{Role}, whose initial value is \emph{Undefined}. By the end of the phase, each particle will have a well-defined Role. Assigning a Role to a particle essentially means telling the particle which element of $S'_F$ it will contribute to forming in the shape composition phase. Once a Leader has selected a set of elements of $S'_F$, it will ``label'' each such element with a unique identifier. Since the number of elements is bounded by a constant, the Leader can memorize the correspondence between identifiers and elements in its internal memory. Then, it will put an identifier in the Role variable of each particle on its directrix, thus effectively assigning each particle an element. The way in which the Leader labels elements is deterministic: so, if the particles have the same representation of $S_F$ in memory, they implicitly agree also on the labeling of the elements of $S'_F$.

The particles that are given the same element identifier will all be contiguous along the directrix: they will form a \emph{chunk}. Chunks can be laid out in any order along a directrix, and can then be moved around and sorted with the general techniques outlined above.

Say that the scale of $S'_F$ is $\lambda$. Then, the sizes of the chunks are as follows:
\begin{itemize}
\item Each chunk corresponding to a super-vertex consists of a single particle.
\item Each chunk corresponding to a super-edge consists of $\lambda-1$ particles.
\item Each chunk corresponding to a super-triangle consists of $(\lambda-1)(\lambda-2)/2$ particles.
\item The chunk corresponding to a partial super-edge consists of $(\lambda-1)/2$ particles. In this case, $k=2$: so $S_F$ must be unbreakably $2$-symmetric, and we assume $\lambda$ to be odd (cf.~Lemma~\ref{l:shapemin}).
\item The chunk corresponding to a partial super-triangle consists of $(\lambda-1)(\lambda-2)/6$ particles. In this case, $k=3$: so $S_F$ must be unbreakably $3$-symmetric, and we assume $\lambda$ not to be a multiple of $3$ (cf.~Lemma~\ref{l:shapemin}).
\end{itemize}

If we fix $S_F$ and we fix $k$, then the Leader has to assign super-vertex identifiers to a constant number $m_v/k$ of particles, super-edge identifiers to $(m_e/k)\cdot(\lambda-1)$ particles (where $m_e/k$ is a constant), super-triangle identifiers to $(m_t/k)\cdot(\lambda-1)(\lambda-2)/2$ particles (where $m_t/k$ is a constant), plus perhaps the identifiers corresponding to one partial super-edge or one partial super-triangle.

For a fixed $S_F$ and a fixed $k$, the number of particles needed, as a function of $\lambda$, is a second-degree polynomial function $\mathcal P(\lambda)=a\lambda^2+b\lambda+c$, for some constants $a$, $b$, $c$ that can be easily computed from $S_F$ as linear expressions of $m_v$, $m_e$, and $m_t$. For instance, $a=m_t/(2k)$ if there is no partial super-triangle among the elements, and $a=m_t/6+1$ if there is a partial super-triangle. In particular, the number of particles that are given a super-vertex identifier or a (partial) super-edge identifier is a linear function $\mathcal P'(\lambda)=b'\lambda+c'$, for some constants $b'$ and $c'$, which are again linear expressions of $m_v$, $m_e$, and $m_t$.

\smallskip
\noindent\textbf{Main role assignment algorithm.}
For the following algorithm to work, we assume the number of particles in the system, $n$, to be large enough compared to the base size of $S_F$, which is a constant $m$. As we will show in Theorem~\ref{tp6}, $n$ has to be at least $\Theta(m^3)$. An explicit multiplicative constant for this $\Theta(m^3)$ bound could be computed from the polynomials $\mathcal P(\lambda)$ and $\mathcal P'(\lambda)$ defined above.

We will focus on a single directrix, and thus we will describe the operations of a single Leader. Of course, all Leaders in the system will do similar operations on their respective directrices. 

The goal of the Leader is to find an appropriate scale $\lambda$ for the shape $S'_F$, equivalent to $S_F$, that the particles will form in the shape composition phase. If $\mathcal P(\lambda)$ is smaller than $n/k$ (i.e., the number of particles on the directrix), then the particles will not be able to fit in $S'_F$; if $\mathcal P(\lambda)$ is too large, then there will not be enough particles to form $S'_F$. Of course, there may not be a $\lambda$ such that $\mathcal P(\lambda)$ is exactly $n/k$, but recall that some particles can be expanded in the final configuration: this gives the system the ability to roughly double the area it can occupy. For this reason, each particle has a flag called \emph{Double}, which is set if and only if the particle is going to be expanded in the final configuration. 

As a first thing, the Leader converts the number $n/k$ in binary, by simply scanning every particle on its directrix and incrementing a binary counter every time it reaches a new particle. This binary number is stored in the first $O(\log n)$ particles on the directrix: each particle remembers one digit.

Then, the Leader sets $\lambda=7$, representing the binary number $7$ in the states of the first three particles on the directrix (these particles will therefore have to remember two binary digits: one for $\lambda$ and one for $n/k$). Then it computes $\mathcal P(\lambda)$, again in binary, with standard multiplication and addition algorithms. The result is again stored in the first particles on the directrix. This number is compared with $n/k$: if $n$ is large enough, we may assume that $n/k>\mathcal P(7)$, and so the computation continues.

Since the current estimate of $\lambda$ is too small, the Leader increments its binary representation by $6$, so as to keep it from being a multiple of $2$ or of $3$, in accordance to Lemma~\ref{l:shapemin} (if the scale of $S'_F$ is a multiple of $k$, then $S'_F$ is not unbreakably $k$-symmetric, and the system is unable to form it). The Leader repeats the above steps on this new $\lambda$, thus computing $\mathcal P(\lambda)$ and comparing it with $n/k$. If $\mathcal P(\lambda)$ is still too small, the Leader increments $\lambda$ by $6$ again, and so on. Observe that the Leader has enough space to compute $\mathcal P(\lambda)$, because it only needs a logarithmic amount of particles, which are abundantly available if $n$ is greater than a (small) constant.

Eventually, the Leader finds the first $\lambda$ such that $\mathcal P(\lambda)\geq n/k$: this will be the final scale of $S'_F$. The Leader also computes $d=\mathcal P(\lambda)-n/k$, which is the number of particles that will have to be expanded in the final configuration. So, it sets the Double flag of the \emph{last} $d$ particles on the directrix: this is equivalent to converting the binary representation of $d$ in unary. As we will show in Theorem~\ref{tp6}, if $n$ is large enough, there are enough particles to complete this operation. As a result, $d$ particles have the Double flag set and $s=n/k-d=2n/k-\mathcal P(\lambda)$ particles do not. Note that $s+2d=\mathcal P(\lambda)$: so, if each Double particles occupies two locations in the final configuration, the system covers an area equal to that of $S'_F$.

Now that the Leader has determined the scale $\lambda$ of $S'_F$, it has to subdivide the particles into chunks and assign Role identifiers to all of them. Note that the Leader still has a binary representation of $\lambda$ stored in the states of the first $O(\log \lambda)$ particles, and so it will be able to use it to count. As explained before, the Leader can easily compute the amount of particles that it has to put in the same chunk, because this is a polynomial function of $\lambda$ that only depends on whether the chunk corresponds to a super-vertex, a (partial) super-edge, or a (partial) super-triangle.

The only thing the Leader has to decide is the order in which to arrange the chunks. It begins by assigning the $m_v/k$ identifiers corresponding to super-vertices of $S'_F$ to the first $m_v/k$ particles on the directrix. Then, if there is a partial super-edge among the elements of $S'_F$, the Leader assigns the corresponding identifier to the next $(\lambda-1)/2$ particles. Then it places all the chunks corresponding to the $m_e/k$ super-edges. If there is a partial super-triangle among the elements of $S'_F$, the Leader places its chunk right after the super-edges. Finally, it places all the chunks corresponding to the $m_t/k$ super-triangles. As $m_v+m_e+m_t$ is bounded by a constant, the Leader can keep track of what identifier it has to assign next by using just its internal memory.

Since the Double particles are the last ones on the directrix, they are more likely to be found in super-triangle chunks. Actually, as we will show in Theorem~\ref{tp6}, if $n$ is large enough, all Double particles will belong to (partial) super-triangle chunks, while the (partial) super-edge chunks and the super-vertex chunks will have no Double particles. The only exception is the case in which $S_F$ has no triangles, and so no element of $S'_F$ is a (partial) super-triangle. In this case, if $n$ is large enough, all Double particles will be in the same super-edge chunk.

\smallskip
\noindent\textbf{Correctness.}
\begin{theorem}\label{tp6}
Let $P$ be the system with $k$ Leader particles resulting from Theorem~\ref{tp5}, and let all particles of $P$ execute the role assignment phase of the algorithm with input a representation of a final shape $S_F$ of constant base size $m$. If $k>1$, we assume that $S_F$ is unbreakably $k$-symmetric. Then, if $n$ is at least $\Theta(m^3)$, there is a stage, reached after $O(n^2)$ rounds, where all particles have a Role identifier. Moreover, if $S_F$ has at least one triangle, then all the Double particles have Role identifiers corresponding to (partial) super-triangles; if $S_F$ consists only of edges, then all the Double particles on the same directrix have a Role identifier corresponding to the same (partial) super-edge. No particle moves in this phase.
\end{theorem}
\begin{proof}
Let us first assume that $S_F$ has some triangles, and so $m_t>0$. Recall that, if $\lambda$ is the scale of $S'_F$ computed by the Leader, then exactly $2n/k-\mathcal P(\lambda)$ particles do not have the Double flag set. We want these particles to include all the super-vertex chunks and the (partial) super-edge chunks, which in turn consist of $\mathcal P'(\lambda)$ particles in total. This is true if and only if $2n/k-\mathcal P(\lambda)\geq \mathcal P'(\lambda)$, which is equivalent to
\begin{equation}\label{eq1}
\frac{2n}k\geq \mathcal P(\lambda)+\mathcal P'(\lambda).
\end{equation}

If $\lambda$ is the scale on which the Leader has stopped, it means that $\lambda-6$ was too small, and hence $n/k>\mathcal P(\lambda-6)$. So,~(\ref{eq1}) reduces to
\begin{equation}\label{eq2}
2\mathcal P(\lambda-6)\geq \mathcal P(\lambda)+\mathcal P'(\lambda).
\end{equation}

Recall that $\mathcal P(\lambda)=a\lambda^2+b\lambda+c$, where the coefficients $a$, $b$, $c$ are linear expressions of $m_v$, $m_e$, and $m_t$. To express this fact, we can write $a=A(m_v, m_e, m_t)$, $b=B(m_v, m_e, m_t)$, and $c=C(m_v, m_e, m_t)$, where $A$, $B$, $C$ are linear functions. So, $\mathcal P = A\cdot\lambda^2 + B\cdot\lambda +C$, where for brevity we have omitted the arguments of $A$, $B$, $C$. Similarly, we have $\mathcal P'(\lambda)=B'\cdot\lambda+C'$.

Note that $\mathcal P(\lambda-6)=A\cdot\lambda^2+B''\cdot\lambda+C''$, where the leading coefficient is the same as the one in $\mathcal P(\lambda)$, and $B''$ and $C''$ are again linear functions of $m_v$, $m_e$, and $m_t$. Therefore,~(\ref{eq2}) becomes
\begin{equation*}
2A\lambda^2 + 2B''\lambda + 2C'' \geq A\lambda^2 + B\lambda + C + B'\lambda + C',
\end{equation*}
or
\begin{equation*}
A\lambda^2 \geq (B+B'-2B'')\cdot \lambda + (C+C'-2C'').
\end{equation*}
Since $\lambda\geq 1$, it suffices to obtain
\begin{equation*}
A\lambda^2 \geq (B+B'-2B''+C+C'-2C'')\cdot\lambda,
\end{equation*}
or
\begin{equation}\label{eq3}
A\lambda \geq B+B'-2B''+C+C'-2C''=D,
\end{equation}
where $D$ is a linear function of $m_v$, $m_e$, and $m_t$.

Since $\lambda$ is the scale that caused the Leader to exhaust the $n/k$ particles on the directrix, we have $\mathcal P(\lambda)\geq n/k$. Since $\lambda\geq 1$, 
\begin{equation*}
(A+B+C)\cdot\lambda^2\geq\mathcal P(\lambda)\geq n/k,
\end{equation*}
and so
\begin{equation*}
\lambda^2\geq \frac n{k\cdot(A+B+C)}.
\end{equation*}

Recall that $A(m_v, m_e, m_t)$ is proportional to $m_t$: it is either $m_t/(2k)$ or $m_t/6+1$, depending on $S_F$. Since we are assuming $m_t>0$, then necessarily $A>0$. Thus, we reduce~(\ref{eq3}) to
\begin{equation*}
\frac{A^2\cdot n}{k\cdot(A+B+C)}\geq D^2,
\end{equation*}
or
\begin{equation*}
n\geq \frac{D^2\cdot k\cdot(A+B+C)}{A^2}.
\end{equation*}
The right-hand side of the above inequality is clearly dominated by a cubic function of $m_v$, $m_e$, and $m_t$, which in turn is dominated by $\Theta(m^3)$. So, if $n\geq \Theta(m^3)$, then (\ref{eq1}) is satisfied.

Suppose now that $S_F$ has no triangles, and so $\mathcal P(\lambda)=\mathcal P'(\lambda)=B\cdot\lambda+C$, by definition of $\mathcal P'$. At the end of the role assignment phase, there are exactly $\mathcal P(\lambda)-n/k$ Double particles, and we want all of them to be in the same (partial) super-edge chunk, which has size at least $(\lambda-1)/2$. Thus, we have to obtain
\begin{equation}\label{eq4}
\frac{\lambda-1}2\geq \mathcal P(\lambda)-\frac nk.
\end{equation}

As before, we have $n/k>\mathcal P(\lambda-6)$, which reduces~(\ref{eq4}) to
\begin{equation}\label{eq5}
\frac{\lambda-1}2\geq \mathcal P(\lambda)-\mathcal P(\lambda-6).
\end{equation}
Observe that $\mathcal P(\lambda)-\mathcal P(\lambda-6)=6B$, and so~(\ref{eq5}) becomes
\begin{equation}\label{eq6}
\lambda\geq 12B+1.
\end{equation}
Again, we have $\mathcal P(\lambda)\geq n/k$. So, $(B+C)\cdot\lambda\geq n/k$, and
\begin{equation*}
\lambda\geq \frac n{k\cdot(B+C)}.
\end{equation*}
This reduces~(\ref{eq6}) to
\begin{equation*}
\frac n{k\cdot(B+C)}\geq 12B+1,
\end{equation*}
or
\begin{equation*}
n\geq (12B+1)\cdot k\cdot(B+C),
\end{equation*}
which is dominated by $\Theta(m^2)$, and \emph{a fortiori} by $\Theta(m^3)$, as required.

It remains to prove the upper bound on the number of rounds. Note that, according to our Turing machine analogy, it does not take more than one round to perform a single step of the machine. Indeed, for the machine to make any progress, the Leader has to be activated; when this happens, the Leader changes its state and sends a message to a neighboring particle in order to transfer the leadership. Then, in at most one round the neighboring particle is activated, so it reads the message, becomes the new Leader, modifies its state, and transfers the leadership to another particle. Therefore, it will suffice to prove that a Turing machine can perform all the required operations in $O(n^2)$ steps.

The first operation is the conversion of $n/k$ (i.e., the length of the ``tape'') in binary, and the naive algorithm works in $O(n^2)$ steps: we scan the tape one cell at a time, and every time we reach a new cell we set a flag in it, we go back to the beginning of a tape, and we increment a binary counter.

Then we have to compute the optimal $\lambda$, starting from $\lambda=7$ and incrementing it by $6$ at each iteration. This process continues until $\mathcal P(\lambda)>n/k$, which means that the binary representation of $\mathcal P(\lambda)$, and therefore that of $\lambda$, takes at most $O(\log n)$ cells. Given a binary representation of $\lambda$, we can compute $\mathcal P(\lambda)$ in $O(\log^2 n)$ machine steps: adding two integers of size $O(\log n)$ takes $O(\log n)$ steps, and multiplying them takes $O(\log^2 n)$ steps (with the usual long multiplication algorithm). Once we have $\mathcal P(\lambda)$, we can compare it with $n/k$ by doing a subtraction, which takes $O(\log n)$ steps.

Since the value of $\mathcal P(\lambda)$ strictly increases every time we increment $\lambda$, we have to repeat the above computations at most $n/k$ times, which takes $O(n\log^2 n) = O(n^2)$ steps overall.

When we have found the correct $\lambda$, we have to set the Double flags of the last $d$ particles, which takes $O(n^2)$ steps with the naive algorithm (similar to the one we used to compute the binary representation of $n/k$).

Finally, we have to assign the Roles to all particles. The size of a chunk is polynomial in $\lambda$, so we can compute the size of all chunks in $O(\log^2 n)$ time overall (since there is a constant number of chunks). Once we have the size of a chunk as a binary number, we use it as a counter to assign a Role to the particles in that chunk. Again, with the naive algorithm this can be done in $O(n^2)$ steps overall.
\end{proof}

\subsection{Shape Composition Phase}\label{s:3.7}
At the end of the role assignment phase, the particles are located on $k$ directrices, each of which has a Leader. The scale $\lambda$ of the final shape $S'_F$, equivalent to the input shape $S_F$, has been determined, and $S'_F$ has been subdivided among the Leaders in equal and symmetric parts. The particles on the same directrix, which are said to be a \emph{team}, have been partitioned into contiguous chunks, each of which corresponds to an element of $S'_F$. In the shape composition phase, the particles will finally form all the elements of $S'_F$. If $k>1$, the particles will actually form a copy of $S'_F$ having center in the center of $S_0$.

\smallskip
\noindent\textbf{Moving to the backbone.}
Recall that, if $k>1$, the elements of $S'_F$ have been split among the Leaders based on their intersections with a structure called backbone (see Section~\ref{s:3.6}). That is, the team that lies on the directrix $\gamma_i$ will form the elements of $S'_F$ that lie on the ray $\beta_i$ of the backbone, as well as other carefully chosen contiguous elements of $S'_F$. Thus, as a preliminary step of the shape composition phase, it is convenient to relocate the whole team from $\gamma_i$ onto $\beta_i$. Of course, if $k=1$, this step is skipped.

Let us consider the case $k=2$ first. In this case, the endpoint of $\beta_i$ is located on $\gamma_i$, at distance $(\lambda-1)/2$ on $G_D$ from its endpoint. To relocate the team, the Leader can reach the last particle on its directrix and execute the pulling procedure introduced in Section~\ref{s:3.5} $(\lambda-1)/2$ times, each time with destination the next vertex along the directrix. Recall that the number $\lambda$ is still represented in binary in the states of the some particles in the team (from the role assignment phase of the algorithm). So, the Leader can easily compute a representation of the number $(\lambda-1)/2$ and use it as a counter to know when to stop pulling (refer to Section~\ref{s:3.6}).

Suppose now that $k=3$. Then, $\beta_i$ is parallel to $\gamma_i$, and its endpoint is at distance $2(\lambda-1)/3$ on $G_D$ from the endpoint of $\gamma_i$. To guide the team to $\beta_i$, the Leader first pulls it along $\gamma_i$ for $(\lambda-1)/3$ steps with the technique explained above. Then it turns counterclockwise by $60^\circ$ and moves in that direction for another $(\lambda-1)/3$ steps, always pulling the entire team. At this point, the Leader is located on $\beta_i$ at distance $n/k-1$ from its endpoint. Finally, the Leader turns counterclockwise by $120^\circ$ and moves in that direction, pulling the team, until the entire line of particles is straight (note that the Leader does not have to count to $n/k-1$ to know when to stop pulling). When these operations are complete, the team is all on $\beta_i$, and the Leader is on its endpoint.

In Theorem~\ref{tfinal}, we will show that we do not have to worry about collisions with particles from other teams during this preliminary step of the algorithm, even if different Leaders end up being completely de-synchronized, and one starts composing $S'_F$ while another is still relocating its own team to the backbone.

\smallskip
\noindent\textbf{Formation order.}
Suppose that, if $k>1$, all the particles on $\gamma_i$ have been relocated to the backbone ray $\beta_i$, and now they are all contracted and form a line segment with an endpoint on the endpoint of $\beta_i$. If $k=1$, we define the unique backbone ray $\beta_1$ to be coincident with the unique directrix $\gamma_1$.

Recall that in the role assignment phase the Leader of $\gamma_i$ has selected some elements of $S'_F$: these constitute a shape $(S'_F)_i\subseteq S'_F$, which the team particles that is now on $\beta_i$ is going to form in the current phase of the algorithm.

We have to decide in what order the elements of $(S'_F)_i$ are to be formed. The super-vertices and the super-edges that lie on $\beta_i$ will be formed last, because $\beta_i$ serves as a ``pathway'' for the team to move and get into position. The other elements of $(S'_F)_i$ are formed starting from the ones adjacent to $\beta_i$, and proceeding incrementally; super-vertices and super-edges are formed first.

This is how the ``ordered list'' $\mathcal L_i$ of elements of $(S'_F)_i$ is constructed:
\begin{itemize}
\item If a (partial) super-edge $e$ of $(S'_F)_i\setminus \beta_i$ is combinatorially adjacent to a super-vertex lying on $\beta_i$ or to a super-vertex that has already been included in $\mathcal L_i$, then $e$ is appended to $\mathcal L_i$.
\item If a super-vertex $v$ of $(S'_F)_i\setminus \beta_i$ is combinatorially adjacent to a super-edge that has already been included in $\mathcal L_i$, then $v$ is appended to $\mathcal L_i$.
\item If all the super-vertices and the (partial) super-edges of $(S'_F)_i\setminus \beta_i$ have already been included in $\mathcal L_i$, then the (partial) super-triangles are appended to $\mathcal L_i$ in any order.
\item If all the elements of $(S'_F)_i\setminus \beta_i$ have already been included in $\mathcal L_i$, then the elements lying on $\beta_i$ are appended to $\mathcal L_i$ in increasing order of distance from the endpoint of $\beta_i$.
\end{itemize}

Once again, we remark that the Leader can store $\mathcal L_i$ in its internal memory, since the number of elements of $(S'_F)_i$ is bounded by a constant.

The list $\mathcal L_i$ does not have to be confused with the order in which the chunks are arranged along the backbone: the chunks can be ordered in any way. Next we are going to show how the Leaders operate to bring the chunks into their right positions and finally form all the elements of $S'_F$.

\smallskip
\noindent\textbf{Main shape composition algorithm.}
The idea of the algorithm is that the $i$th Leader guides its team in the formation of $(S'_F)_i\subseteq S_F$, one element at a time, following the list $\mathcal L_i$. At any time, the particles of the team that do not lie in $(S'_F)_i$ are all lined up on $\beta_i$, and constitute a ``repository'' of contiguous chunks, each with a Role identifier corresponding to an element of $(S'_F)_i$.

In this discussion, we will temporarily forget about the presence of Double particles in the repository. The formation of elements by chunks containing Double particles will be treated after the main parts of the algorithm have been explained. Other details of the algorithm will be covered later, as well.

The following steps are executed assuming that the $i$th Leader is on $\beta_i$, within the repository, and are repeated until there are no more elements on $\mathcal L_i$ to form. The repository is assumed to consist of contracted particles forming a connected sub-segment of $\beta_i$; moreover, the part of $\beta_i$ that follows the repository is assumed to be devoid of particles. These conditions are satisfied when the algorithm begins and will be satisfied again every time the steps have been executed.
\begin{itemize}
\item The Leader reads the identifier of the next element $d$ on the list $\mathcal L_i$ (i.e., the identifier of the first element on the list that has not been formed, yet).
\item The Leader locates the particles in the repository that have Role identifier corresponding to $d$ (in our terminology, these particles constitute a chunk) and ``shifts'' them to the beginning of the repository (i.e., the part of the repository that is closest to the endpoint of $\beta_i$). That is, the Leader swaps their Role identifiers and Double flags with the ones of the particles that precede them, until the desired particles are at the beginning. This operation is simple to do, considering the Turing machine analogy pointed out in Section~\ref{s:3.6}. Note that the particles do not have to physically move, but only exchange messages and modify their internal states.
\item Suppose that $d$ does not lie on $\beta_i$. Since $d$ is next on the list, it means that there is a sequence of elements of $(S'_F)_i$ connecting $\beta_i$ with $d$ that have already been formed (refer to the definition of $\mathcal L_i$). More precisely, there is a sequence $W=(v_0, e_0, v_1, e_1, \dots)$, where the $v_j$'s are super-vertices and the $e_j$'s are super-edges, and each element is combinatorially adjacent to the next, such that $v_0$ lies on $\beta_i$, all elements of $W$ except $v_0$ have already been formed, and the last element $d'$ (which could be a super-vertex or a super-edge) is combinatorially adjacent to $d$. Note that $W$ induces a path in $(S'_F)_i$, because it consists of super-edges and the super-vertices between them. Let $q$ be the first point along this path that neighbors a point of $d$.

In the special case in which $d$ is a super-triangle combinatorially adjacent to a super-edge $e$ lying on $\beta_i$, the above does not hold. In this case, we take both $v_0$ and $q$ to be the same endpoint of $e$.

Then the following steps are performed:
\begin{itemize}
\item The Leader pulls the entire repository along $\beta_i$ until the particle that is closest to the endpoint of $\beta_i$ coincides with $v_0$ (we will explain how the Leader can find $v_0$ later). No ``obstructions'' are found on $\beta_i$, because at this stage it does not contain formed elements, yet.
\item The Leader shifts along $W$ all the particles of the chunk corresponding to $d$, in such a way that the first particle of the chunk goes from $v_0$ to $q$ (and the other particles of the chunk occupy the positions on $W$ before $q$, and perhaps also on $\beta$, if the chunk is too long). As a consequence, all the particles of $W$ are shifted back along $W$ and into $\beta$ by as many positions as the size of the chunk.
\item The Leader pulls the chunk into $d$, along with all of $W$ and the rest of the repository. As a result, the chunk forms $d$ and the particles of $W$ are back into their original positions (i.e., the ones they occupied in the previous step before being shifted). The details of how the Leader arranges the chunk to form $d$, in case $d$ is a (partial) super-triangle, will be explained later.
\item The Leader returns to $v_0$ along $d$ and $W$.
\end{itemize}
\item Suppose now that $d$ lies on $\beta_i$. The following steps are performed:
\begin{itemize}
\item The Leader pulls the entire repository along $\beta_i$ until the particle that is closest to the endpoint of $\beta_i$ coincides with the vertex of $d$ that is farthest from the endpoint of $\beta_i$. Since the elements of $(S'_F)_i$ lying on $\beta_i$ have been inserted in $\mathcal L_i$ in order of distance from the endpoint of $\beta_i$, there are no particles on $\beta_i$ ``obstructing'' the repository while it is being pulled. Again, the details of how the Leader finds this point on $\beta_i$ will be explained later.
\item If $d$ is a super-vertex, it has already been formed. Otherwise, $d$ is a super-edge: the Leader forms it by pulling the entire repository toward the endpoint of $\beta_i$ for $\lambda-2$ times. As usual, the Leader can easily count to $\lambda-2$ using the states of the $\lambda-1$ particles in the chunk, and therefore it knows when to stop pulling.
\end{itemize}
\end{itemize}

The Leader always makes sure to keep a binary representation of $\lambda$ in the repository for as long as possible (i.e., as long as there are enough particles in the repository). So, before removing a chunk that contains part of this information, the Leader copies it to other chunks.

\smallskip
\noindent\textbf{Traveling long distances on the backbone.}
In the algorithm above, the Leader is supposed to pull the repository along $\beta_i$ until it reaches a certain element of $(S'_F)_i$, which can be far away. We have to explain how this element can be found, considering that the Leader may not be able to measure this distance by counting. We can assume this element to be a super-vertex: if it is a super-edge lying on $\beta_i$, then $\beta_i$ also contains the two super-vertices that bound it, and the Leader may as well reach one of those instead.

Observe that, as the Leader executes the above steps, it always knows in which element of $(S'_F)_i$ it is located, because it can keep track of it using only a constant amount of memory.

Now, suppose that the Leader is located on a super-vertex $u_1$ on $\beta_i$ and has to move to another super-vertex $u_2$, always on $\beta_i$, while pulling the repository. Obviously, the distance between $u_1$ and $u_2$ is $\lambda$ times the distance between the corresponding vertices in the minimal shape $S_F$, which in turn is a known value that is bounded by a constant (because the base size of $S_F$ is a constant). It follows that the Leader can measure this distance if it can count to $\lambda$.

If the current repository contains a chunk corresponding to a (partial) super-edge or a (partial) super-triangle, then the Leader has enough particles at its disposal to count to $\lambda$ in binary, and the problem is solved.

So, let us study the case in which the current repository only contains chunks corresponding to super-vertices. We deduce that all the super-edges of $(S'_F)_i$ have already been formed, and the Leader has to reach $u_2$ to form a super-vertex $v$.

Suppose first that $v$ is not on $\beta_i$. Then, there is a path consisting of super-edges and super-vertices of $(S'_F)_i$ that connects $u_2$ with $v$ (due to the way the elements of $(S'_F)_i$ have been selected by the Leader; see Section~\ref{s:3.6}). In particular, there is an edge $e$, combinatorially adjacent to $u_2$, that has already been formed. So, the Leader can simply proceed along $\beta_i$ until it finds a particle with Role identifier corresponding to $e$ among its neighbors. When it finds such a particle, the Leader is in $u_2$.

Now suppose that $v$ is on $\beta_i$, and so $u_2=v$. Recall from Section~\ref{s:3.6} that, if the Leader has selected $u_2$ to be part of $(S'_F)_i$, then it has also selected a (partial) super-edge $e$ that is combinatorially adjacent to it. Again, since $e$ has already been formed, the Leader can proceed along $\beta_i$ until it finds a neighbor with Role identifier corresponding to $e$.

\smallskip
\noindent\textbf{Forming shapes with no triangles.}
So far, we have ignored the presence of Double particles, i.e., particles that have to be expanded in the final configuration. Next we will explain how to handle them in the shape composition algorithm.

The issues arise from the fact that a Leader has to be able to move through contracted particles to reach different elements of $(S'_F)_i$, as well as pull chains of particles that are supposed to be contracted. In the shape composition algorithm, this happens when the Leader has to form an element that is reachable from $\beta_i$ through a path $W$ consisting of super-vertices and super-edges. If $W$ contains Double particles, which are either expanded or leave gaps between particles, then these operations are not straightforward.

We first consider the case in which $S_F$ has no triangles and consists only of edges. Recall from Theorem~\ref{tp6} that, in this case, all the Double particles are in a single chunk $c$ that corresponds to a (partial) super-edge $e$. Note that which chunk actually contains the Double particles is irrelevant for the purposes of shape formation, and so we can choose to put them in a convenient chunk.

The chunk $c'$ we choose is the one corresponding to the super-edge $e'$ that is last in the list $\mathcal L_i$. To do the switch, the Leader simply checks the Role identifiers of all the particles in the repository, and changes each occurrence of the identifier corresponding to $e$ into the one corresponding to $e'$, and vice versa.

The advantage of choosing $e'$ is that it will never be part of a path $W$ that the Leader has to follow to reach the next element to form, except perhaps if such an element is a specific super-vertex $v'$ that is combinatorially adjacent to $e'$. So, if $v'$ appears in the list $\mathcal L_i$ after $e'$, the Leader first sees if it can move $v'$ before $e'$ while respecting the constraints that define $\mathcal L_i$. This is possible if and only if $(S'_F)_i$ has a super-edge distinct from $e'$ that is combinatorially adjacent to $v'$. If this is not the case, the Leader merges the two chunks corresponding to $e'$ and $v'$, and hence it will form $e'$ and $v'$ in a single step, as if they were a slightly longer super-edge.

So, all the elements of $(S'_F)_i$ except $e'$ and perhaps $v'$ are formed as explained in the main shape composition algorithm. Indeed, if $v'$ has to be part of a path $W$ as defined above, then it means that $v'$ has a combinatorially adjacent super-edge in $(S'_F)_i$ other than $e'$, and so $v'$ has been formed as normal, and its chunk has not been merged with the one of $e'$.

We just have to show how to form $e'$, and perhaps $v'$ if it is formed in the same step. The algorithm works as normal, until the chunk that is going to form $e'$ (and perhaps $v'$) has been pulled up to a point where one of its particles neighbors an endpoint of $e'$. Then, the following steps are executed:
\begin{itemize}
\item The Leader pulls the chunk along $e'$ (also pulling the path $W$ and the repository, as usual) until the last particle of the chunk has entered $e'$.
\item The Leader makes the particle on which it currently is into a \emph{Puller} (by setting an internal flag).
\item The Leader leaves $e'$ and proceeds with the algorithm as normal. The Puller waits for the Leader to leave $e'$.
\item The Puller starts another pulling procedure. When the last Double particle of $e'$ is pulled, it expands and sends a Movement-Done message to its Predecessor without contracting again.
\item The above step is repeated until all the Double particles of the chunk are expanded.
\end{itemize}

\smallskip
\noindent\textbf{Forming shapes with triangles.}
Finally, we consider the case in which $S_F$ has at least one triangle. According to Theorem~\ref{tp6}, in this case all the Double particles are in chunks corresponding to (partial) super-triangles. So, all the paths consisting of super-vertices and super-edges that the Leader has to traverse to reach new elements of $(S'_F)_i$ are unaffected by Double particles. Therefore, the shape composition algorithm works as we already explained, except for the formation of (partial) super-triangles.

In the following, we will explain how to form a (partial) super-triangle $t$ whose corresponding chunk $c$ may contain Double particles. Recall that $c$ enters $t$ through an endpoint of a combinatorially adjacent (partial) super-edge, and therefore it starts covering $t$ from one of its three corners.

Suppose first that $t$ is a super-triangle. The algorithm is roughly the same as the one already used above for the super-edges, except that now the Leader has to fill a triangle $t$. The steps are as follows (refer to Figure~\ref{f:triangle}):
\begin{itemize}
\item The Leader pulls $c$ (as well as the path $W$ and the repository) following the boundary of $t$ in the counterclockwise direction for $\lambda-2$ steps, thus covering one of its sides. The Leader can count to $\lambda-2$ as usual, representing $\lambda-2$ in the states of the particles of $c$.
\item The Leader turns counterclockwise by $120^\circ$ and pulls for another $\lambda-3$ steps, covering another side of $t$.
\item The Leader turns counterclockwise by $120^\circ$ and pulls until it finds a particle in front of it.
\item The previous step is repeated until the last particle of $c$ enters $t$.
\item The Leader makes the particle on which it currently is into a Puller.
\item The Leader leaves $t$ (following the chain it just pulled) and proceeds with the shape composition algorithm. The Puller waits for the Leader to leave $t$.
\item The Puller pulls $c$, turning counterclockwise by $120^\circ$ if it finds a particle in front of it. When the last Double particle of $c$ is pulled, it expands and sends a Movement-Done message to its Predecessor without contracting again.
\item The above step is repeated until all the Double particles of the $c$ are expanded.
\end{itemize}

\begin{figure}[ht!]
\begin{center}
  \subfloat[The chunk is pulled along the arrow.]{
  \includegraphics[width=0.65\textwidth]{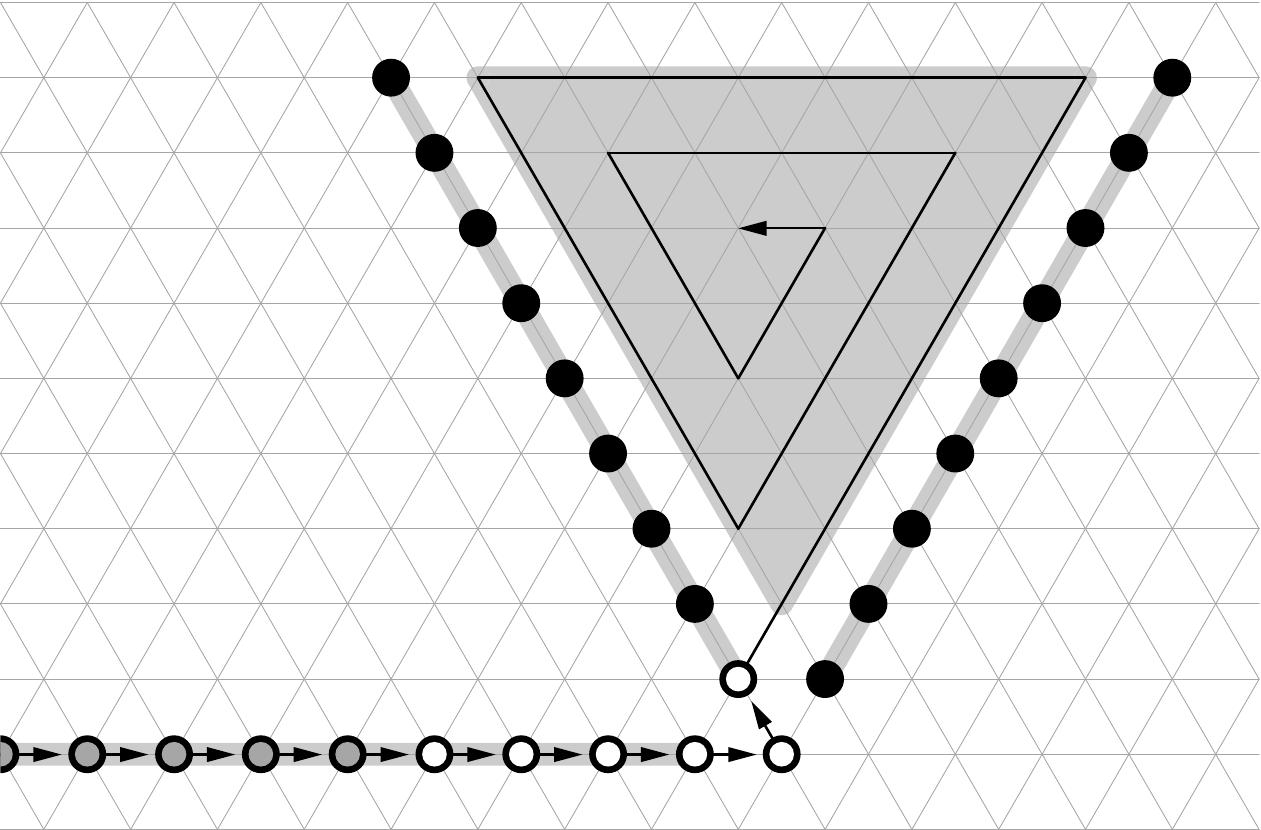}
  \label{f:triangle1}
  }
  \\
  \subfloat[The Double particles expand to cover the super-triangle.]{
  \includegraphics[width=0.65\textwidth]{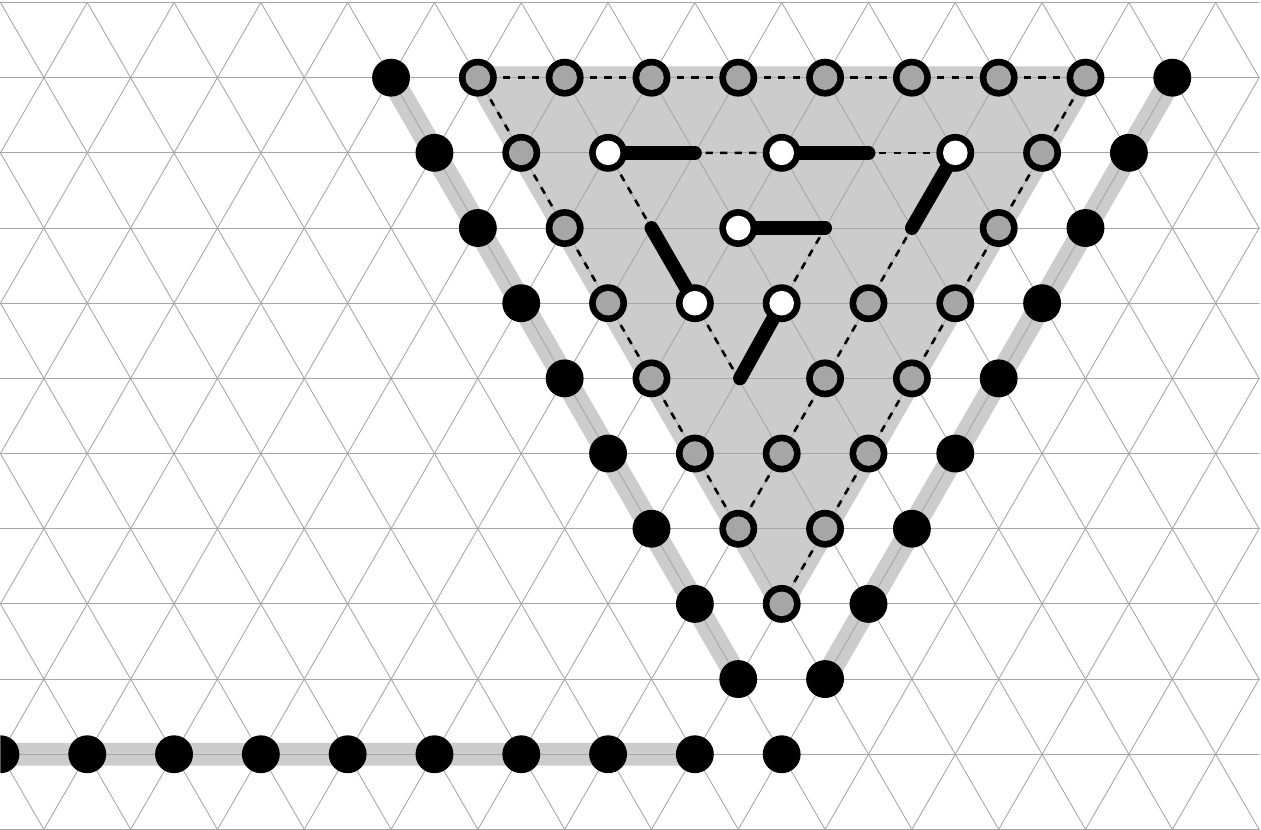}
  \label{f:triangle2}
  }
\end{center}
\caption{Formation of a super-triangle by a chunk containing some Double particles. The Double particles are drawn in white; the other particles in the chunk are drawn in gray.}
\label{f:triangle}
\end{figure}

If $t$ is a partial super-triangle, the algorithm is identical, with the only difference that now the Leader does not have to cover the perimeter of an equilateral triangle with sides of length $\lambda-3$, but of an isosceles trapezoid with sides of length $(\lambda-4)/3$, $(\lambda-4)/3$, $(\lambda-4)/3$, and $2(\lambda-4)/3$.

\smallskip
\noindent\textbf{Correctness.}
We can now prove the correctness of the universal shape formation algorithm.
\begin{theorem}\label{tfinal}
Let $P$ be a system of $n$ particles forming a simply connected shape $S_0$ at stage 0. Let $S_F$ be a shape of constant base size $m$ that is unbreakably $k$-symmetric if $S_0$ is unbreakably $k$-symmetric. If all particles of $P$ execute the universal shape formation algorithm with input a representation of the final shape $S_F$, and if $n$ is at least $\Theta(m^3)$, then there is a stage, reached after $O(n^2)$ rounds, where $P$ forms a shape equivalent to $S_F$. The total number of moves performed by $P$ up to this stage is $O(n^2)$, which is asymptotically optimal; particles no longer move afterwards. 
\end{theorem}
\begin{proof}
Under these assumptions, Theorems~\ref{tp1}--\ref{tp6} apply. So, we can assume that at some stage the particles will be found on $k$ directrices, each containing exactly one Leader particle. Moreover, if $k>1$, then $S_F$ is unbreakably $k$-symmetric. The Leaders have implicitly agreed on a shape $S'_F$ and have split its elements among each other. Since the handedness on which the particles agree (cf.~Theorem~\ref{tp3}) may not be the ``real'' one, $S'_F$ may actually be a reflected copy of $S_F$. However, the definition of the shape formation problem allows for any similarity transformation of the shape, which include reflections (see Section~\ref{s:2}).

If $k=1$, the correctness of the shape composition algorithm follows by construction. If there is more than one Leader, we only have to show that different Leaders will never interfere with each other, and their respective teams will never get in each other's way. This is because different teams are confined to move within different regions of $G_D$ throughout the phase. This is obvious if the $k$ teams are all executing the preliminary relocation step or if they are all executing the main composition algorithm. Suppose now that $k=3$, and one team is moving from its directrix $\gamma_i$ to the backbone $\beta_i$, while another team is already executing the main composition algorithm. Recall from Section~\ref{s:3.6} that all the elements of $S'_F$ that are incident to the co-backbone ray $\beta'_i$ are selected by the Leader of $\gamma_i$ to be part of $(S'_F)_i$: these are precisely the elements that are incident with $\gamma_i$, as well. So, as particles are being pulled from $\gamma_i$ to $\beta_i$, they only pass through elements that have been selected by their Leader, making it impossible for them to encounter particles from other teams. 

Let us count the total number of moves of $P$ and the number of rounds it takes to form $S'_F$. Up to the beginning of the shape composition phase, $P$ performs at most $O(n^2)$ moves in at most $O(n^2)$ rounds, as Theorems~\ref{tp1}--\ref{tp6} imply. When a Leader relocates its team onto the backbone, it pulls all the particles at most $O(n)$ times, and the total number of moves, as well as rounds, is at most $O(n^2)$. Then, in order to form one element of $S'_F$, a Leader may have to pull at most $O(n)$ particles for at most $O(n)$ times along the backbone to get the chunk into position: this yields at most $O(n^2)$ moves and rounds. Then it has to pull at most $O(n)$ particles for a number of times that is equal to the size of the element of $S'_F$, which is $O(n)$. Since the number of elements of $S'_F$ is bounded by a constant, this amounts to at most $O(n^2)$ moves and rounds, again. All other operations involve only message exchanges and no movements, so the $O(n^2)$ upper bound on the number of moves follows. Due to the matching lower bound given by Theorem~\ref{t:lower}, our universal shape formation algorithm is asymptotically optimal with respect to the number of moves.

To conclude, observe that shifting chunks within a repository, computing polynomial functions of $\lambda$, and using them as counters takes $O(n^2)$ rounds overall, since this has to be done at most once per chunk, i.e., a constant number of times. So, the upper bound of $O(n^2)$ rounds follows, as well.
\end{proof}

\section{Conclusion and Further Work}\label{s:4}
We have described a universal shape formation algorithm for systems of particles that performs at most $O(n^2)$ moves, which is asymptotically optimal. The number of rounds taken to form the shape is $O(n^2)$ as well: with a slight improvement on the last phases of out algorithm, we can reduce it to $O(n\log n)$ rounds, and the example described in Theorem~\ref{t:lower} yields a lower bound of $\Omega(n)$ rounds. Determining an asymptotically optimal bound on the number of rounds is left as an open problem.

We have established that, given a shape $S_F$ of constant size $m$, a system of $n$ particles can form a shape geometrically similar to $S_F$ (i.e., essentially a scaled-up copy of $S_F$) starting from any simply connected configuration $S_0$, provided that $S_F$ is unbreakably $k$-symmetric if $S_0$ is, and provided that $n$ is large enough compared to $m$. We only determined a bound of $\Theta(m^3)$ for the minimum $n$ that guarantees the formability of $S_F$. We could improve it to $\Theta(m)$ by letting the Double particles be in any chunk and adopting a slightly more sophisticated pulling procedure in the last phase. We may wonder if this modification would make our bound asymptotically optimal. 

When discussing the role assignment phase, when the particles are arranged along straight lines, we have argued that the system can compute any predicate that is computable by a Turing machine on a tape of limited length. If we allow the particles to move back and forth along these lines to simulate registers, we only need a (small) constant number of particles to implement a full-fledged Turing machine with an infinite tape. So, in the role assignment phase, we are actually able to compute any Turing-computable predicate (although we would have to give up our upper bounds of $O(n^2)$ moves and stages).

With this technique, we are not only able to replace our $\Theta(m^3)$ with the best possible asymptotic bound in terms of $m$, but we have a universal shape formation algorithm that, for every $n$ and every $S_F$, lets the system determine if $n$ particles are enough to form a shape geometrically similar to $S_F$. This is done by examining all the possible connected configurations of $n$ particles and searching for one that matches $S_F$, which is of course a Turing-computable task.

Taking this idea even further, we can extend our notion of shape to its most general form. Recall that the shapes considered in~\cite{spaa} were sets of ``full'' triangles: when a shape is scaled up, all its triangles are scaled up and become larger full triangles. In this paper, we extended the notion of shape to sets of full triangles and edges: when an edge is scaled up, it remains a row of points. Of course, we can think of shapes that are not modeled by full triangles or edges, but behave like fractals when scaled up. For instance, we may want to include a discretized version of the Sierpinski triangle as a ``building block'' of our shapes, alongside full triangles and edges. Scaling up a discretized Sierpinski triangle is equivalent to increasing its ``resolution'', which causes finer details to appear inside it. Clearly, these scaled-up copies of the discretized Sierpinski triangle are Turing-computable.

Generalizing, we can replace our usual notion of geometric similarity between shapes with any Turing-computable equivalence relation $\sim$. Then, the shape formation problem, with input a shape $S_F$, asks to form any shape $S'_F$ such that $S_F\sim S'_F$. This definition of shape formation problem includes and greatly generalizes the one studied in this paper, and even applies to scenarios that are not of a geometric nature. Nonetheless, this generalized problem is still solvable by particles, thanks to the technique outlined above.

\smallskip
\noindent\textbf{Acknowledgments.}
This research has been supported in part by the Natural Sciences and Engineering Research Council of Canada under the Discovery Grant program, by Prof.~Flocchini's University Research Chair, and by Prof.~Yamauchi's Grant-in-Aid for Scientific Research on Innovative Areas ``Molecular Robotics'' (No.\ 15H00821) of MEXT, Japan and JSPS KAKENHI Grant No.\ JP15K15938.

\bibliographystyle{plain}

\end{document}